\documentclass[11pt,letterpaper]{article}
\usepackage{amsthm,amsmath,amssymb,amsfonts} 
\usepackage{epsfig} 
\usepackage{latexsym,nicefrac,bbm}
\usepackage{xspace}
\usepackage{color,fancybox,graphicx,subfigure,fullpage}
\usepackage[top=1in, bottom=1in, left=1in, right=1in]{geometry}
\usepackage{tabularx}
\usepackage{hyperref} 
\usepackage{pdfsync}
\usepackage{multicol}
\usepackage{cite,cleveref}

\usepackage{float}
\renewcommand{\epsilon}{\varepsilon}

\usepackage{epsfig}
\usepackage{verbatim}

\theoremstyle{definition}
\newtheorem{problem}{Problem}
\newtheorem{assumption}{Assumption}

\usepackage{algorithmicx}
\usepackage{algorithm,caption}
\usepackage{algpseudocode}

\usepackage{mhequ}
\def \be{\begin{equs}}
\def \ee{\end{equs}}

\newtheorem{theorem}{Theorem}[section]
\newtheorem{lemma}[theorem]{Lemma}
\newtheorem{definition}[theorem]{Definition}

\newtheorem{proposition}[theorem]{Proposition}

\newtheorem*{theorem*}{Theorem}

\crefname{theorem}{Theorem}{Theorems}
\crefname{observation}{Observation}{Observations}
\crefname{proposition}{Proposition}{Propositions}
\crefname{claim}{Claim}{Claims}
\crefname{condition}{Condition}{Conditions}
\crefname{example}{Example}{Examples}
\crefname{fact}{Fact}{Facts}
\crefname{lemma}{Lemma}{Lemmas}
\crefname{corollary}{Corollary}{Corollaries}
\crefname{definition}{Definition}{Definitions}
\crefname{remark}{Remark}{Remarks}

\title{\bf Convex Optimization with Unbounded Nonconvex Oracles using Simulated Annealing}

\usepackage{authblk}

\author[1]{ Oren Mangoubi}
\author[2]{Nisheeth K. Vishnoi}
\affil[1,2]{\small \'{E}cole Polytechnique F\'{e}d\'{e}rale de Lausanne (EPFL), Switzerland}

\begin{document}

\maketitle

\begin{abstract}
We consider the problem of minimizing a convex objective function $F$ when one can only evaluate its noisy approximation $\hat{F}$.
Unless one assumes some structure on the noise, $\hat{F}$ may be an arbitrary nonconvex function, making the task of minimizing $F$ intractable.
To overcome this, prior work has often focused on the case when $F(x)-\hat{F}(x)$ is uniformly-bounded.
In this paper we study the more general case when the noise has magnitude $\alpha F(x) + \beta$ for some $\alpha, \beta > 0$, and present a polynomial time algorithm that finds an approximate minimizer of $F$ for this noise model.
Previously, Markov chains, such as the stochastic gradient Langevin dynamics, have been used to arrive at approximate solutions to these optimization problems.
However, for the noise model considered in this paper, no single temperature allows such a Markov chain to both mix quickly and concentrate near the global minimizer.
We bypass this by combining ``simulated annealing" with the stochastic gradient Langevin dynamics, and gradually decreasing the temperature of the chain in order to approach the global minimizer.
 As a corollary one can approximately minimize a nonconvex function that is close to a convex function; however, the closeness can deteriorate as one moves away from the optimum.
\end{abstract}

\section{Introduction}
A general problem that arises in machine learning, computational mathematics and optimization is that of minimizing a convex objective function $F:\mathcal{K} \rightarrow \mathbb{R}$, where $\mathcal{K} \subseteq \mathbb{R}^d$ is convex, and one can only evaluate $F$ {\em approximately}.
Let $\hat{F}$ denote this ``noisy'' approximation to $F$.
In this setting, even though the function $F$ is convex, we can no longer assume that $\hat{F}$ is convex.
However, if one does not make any assumption on the noise function, the problem of minimizing $F$ can be shown to be arbitrarily hard.
Thus,  having some restrictions on the noise function is necessary.

A well studied setting is that of  ``additively'' bounded noise \cite{applegate_kannan,singer2015information,risteski2016algorithms, hitting_times}.
Here, the noise $N(x):=F(x)-\hat{F}(x)$ is assumed to have  a uniform bound on $\mathcal{K}$: $\sup_{x \in \mathcal{K}} |N(x)| \leq \beta$ for some $\beta \geq 0$.
In practice, however, the strongest bound we might have for the noise may not be uniform on $\mathcal{K}$.     
One such noise model is that of ``multiplicative''  noise  where one assumes that  $|\hat{F}(x) - F(x)| \leq \alpha F(x),$ for all $x\in  \mathcal{K}$  and some $\alpha \geq 0$.
{In other words, $|N(x)| = |\hat{F}(x) - F(x)| = |F(x)|\times \alpha$, which motivates the name.    %
 One situation where multiplicative noise arises is when $F$ decomposes into a sum of  functions that are easier to compute, but these component functions are computed via Monte Carlo integration and the stopping criteria of these integration methods  depend on the computed value of the component function 
 \cite{chen_thesis}.
For other natural settings where multiplicative noise arises see \cite{chen2015stochastic,jebalia2008multiplicative,jebalia2011log}.
{More generally, one can model the noisy function $\hat{F}$ by decomposing it into additive and multiplicative components, in the following sense:
 \begin{definition} \label{def:noise_both2}
We say that $\hat{F}$ has additive and multiplicative noise  levels $(\beta, \alpha)$ if
\be \label{eq:model_add_mult2}
|\hat{F}(x)- F(x)| \leq \alpha(F(x) -F(x^{*})) + \beta.
\ee
\end{definition}}
\noindent
Note that this noise model has the natural property that the noise level does not change if we replace $F(x)$ and $\hat{F}(x)$ with a new objective function $F(x) + C$ and a new oracle $\hat{F}(x) + C$ for the same number $C>0$.
To motivate this definition, we consider a situation where the error in computing the objective function depends on the amount of time spent on the computation.  For instance, to compute the objective function one may be required to solve a complicated system of partial differential equations, where a finer discretization leads to a greater accuracy but also to a longer computation time \cite{conrad2018parallel, cliffe2011multilevel}.  In this case, one can start by using a short computation time for each evaluation and gradually increase the computation time as one approaches the minimum value.  Whereas a purely additive noise model would require one to have a uniform computational cost at each step, the multiplicative noise model allows one to analyze methods where one has the flexibility to use a different cost at each step.

As another application, we consider the problem of solving a system of noisy linear or nonlinear black-box equations where one wishes to find a value of $x$ such that $h_i(x)=0$ for each component function $h_i$ \cite{chen2015stochastic}. 
 Since each equation $h_i(x)=0$  must be satisfied simultaneously for a single value of $x$, it is not enough to solve each equation individually.  
One way in which we may solve this system of equations is by minimizing an objective function of the form
$F(x) = \frac{1}{n} \sum_{i=1}^n(h_i(x))^2$
since any value of $x$ that minimizes $F(x)$ also solves the system of equations $h_i(x)=0$ for every $i$, provided that such a solution exists.  While it is true that one may instead minimize the objective function  $\frac{1}{n} \sum_{i=1}^n |h_i(x)|$ to solve the same system of equations, it is oftentimes preferable to use the quadratic objective function  $F(x) = \frac{1}{n} \sum_{i=1}^n(h_i(x))^2$ since it is much smoother and can lead to faster convergence in practice \cite{chen2015stochastic}.
Rather than having access to an exact computation oracle for $h_i(x)$ one may instead only have access to a perturbed function $\hat{h}_i(x) = h_i(x) + N_i(x).$
  Here $N_i(x)$ is a noise term that may have additive or multiplicative noise (or both), that is, $|N_i(x)| \leq b + ah_i(x)$ for some $a,b \geq 0$. 
Hence, instead of minimizing the objective function $F$, one must try to minimize a noisy function of the form 
$\hat{F}(x) =\frac{1}{n} \sum_{i=1}^n(\hat{h}_i(x))^2.$
A straightforward calculation shows that the fact that $|N_i(x)| \leq b + ah_i(x)$ for all $i$ implies that 
\be
|\hat{F}(x) - F(x)| 
 \leq  (2a + a^2 + 2b+2ab)(F(x) - F(x^{*})) + \frac{1}{2}(b+ab) + b^2,
\ee
where $F(x^{*}) = 0$.  Thus, $\hat{F}$ can be modeled as having additive noise  level $\beta = \frac{1}{2}(b+ab) + b^2$ together with multiplicative noise  level $\alpha = 2a + a^2 + 2b+2ab$. 
  In particular, even if each component function only has additive noise (that is, if $a=0$), $\hat{F}$ will still have nonzero multiplicative noise $\alpha = 2b$.
Thus we arrive at  the following  general problem.
\begin{problem}\label{problem1}
Let $\mathcal{K} \subseteq \mathbb{R}^d$ be a convex body and $F: \mathcal{K} \rightarrow \mathbb{R}$ be a  convex function, where $F(x^{*})=0$
and $x^\star$ is a minimizer of $F$ in $\mathcal{K}$.
Given access to a noisy oracle $\hat{F}$ for $F$ that has additive and multiplicative noise  levels $(\beta,\alpha)$.  
The problem is to find an approximate minimizer $\hat{x}$ for $F$ such that $F(\hat{x}) \leq \hat{\varepsilon}$ for a given $\hat{\varepsilon} >0$.
\end{problem}

\noindent
One of the first papers to study this problem was  by  \cite{applegate_kannan} in the special case of additive noise (where $\alpha=0$).  
Specifically, they studied the related problem of sampling from the canonical distribution $ \frac{1}{\int_{\mathcal{K}}e^{-\xi \hat{F}(y)}\mathrm{d}y}e^{-\xi \hat{F}(x)}$ when $\hat{F}$ is an additively noisy version of a convex function. 
Roughly, their algorithm discretized $\mathcal{K}$ with a grid and ran a simple random walk on this grid.  
Using their Markov chain one can solve Problem \ref{problem1} in the special case of $\alpha = 0$ for some error $\hat{\varepsilon} =\tilde{O}(d\beta)$ with running time that is polynomial in $d$  and various other parameters as well.

In \cite{Simulated_Annealing_Nonassymptotic},  Problem \ref{problem1} was studied in a  special case where the noise decreases to zero near the global minimum\footnote{\cite{Simulated_Annealing_Nonassymptotic} also study separately the special case of purely additive noise, but not simultaneously in the presence of a non-uniformly bounded noise component.} and $F$ is $\mathfrak{m}$-strongly convex.
Specifically, they study the situation where the noise is bounded by $|N(x)| \leq c\|x-x^{\star}\|^p$, for some $0<p<2$ and some $c>0$.
Roughly speaking, in this regime they show that one can obtain an approximate minimizer $\hat{x}$ such that $F(\hat{x}) - F(x^\star) \leq O((\frac{d}{\mathfrak{m}})^{\frac{1}{2-p}})$ in polynomial time. 
 To find an approximate minimizer $\hat{x}$, they repeatedly run a simulated annealing Markov chain based on the ``hit-and-run'' algorithm.
 They state that they are ``not aware of optimization methods for such a problem" outside of their work, and that ``it is rather surprising that one may obtain provable guarantees through simulated annealing" under noise with non-uniform bounds even in the special case of strong convexity.

Problem \ref{problem1} was also studied by \cite{hitting_times}  in the special case of  additive noise (where $\alpha=0$ but $\beta \geq 0$).
The main component of their algorithm is the  stochastic gradient Langevin dynamics (SGLD) Markov chain that runs at a fixed ``temperature" parameter to find an approximate minimizer $\hat{x}$.  
In particular, they show that one can solve Problem \ref{problem1} in the special case of $\alpha = 0$ for some error $\hat{\varepsilon} =\tilde{O}(d\beta)$ with running time that is polynomial in $d$ and $\beta$ and various smoothness parameters.
 Other related works that have studied various aspects of optimization under additive noise  include  \cite{singer2015information,hazan2016graduated, risteski2016algorithms}.

The difficulty of extending these results to the general case when both $\alpha,\beta >0$, and $F$ is not necessarily strongly convex arises from the fact that, in this setting, the noise can become unbounded and the prior Markov chain  approaches do not seem to work. 
Roughly, the Markov chains of \cite{applegate_kannan,hitting_times} run at a fixed temperature and, due to the fact that the noise can be very different at different levels of $F$, would either get stuck in  a local minimum or never come close to the minimzer; see Figure \ref{fig:annealing} for an illustration.
The Markov chain of \cite{Simulated_Annealing_Nonassymptotic} on the other hand varies the temperature but the strong convexity of $F$ makes the task of estimating progress significantly simpler.

\subsection{Our contributions}
The main result of this paper is the first polynomial time algorithm that solves Problem \ref{problem1} when $\alpha,\beta>0$ without assuming that $F$ is strongly convex.
Our algorithm combines simulated annealing (as in \cite{Simulated_Annealing_Nonassymptotic}) with the stochastic gradient Langevin dynamics (as in \cite{hitting_times}).
We assume  that $\|\nabla F\| \leq \lambda$ and  that $\mathcal{K}$ is contained in a bounding ball of radius $R>0$, and that $\mathcal{K}= \mathcal{K}' + B(0,r')$ for some $r'>0$, where ``+" denotes the Minkowski sum.
Note that, given bounds  $\lambda$ and $R$, one can deduce an upper bound of $\lambda R$ on the value of  $F$ in $\mathcal{K}$. 
Also note that while the Lipschitz gradient assumption helps us prove running time bounds for our algorithm, it is likely not needed to solve the problem.

\begin{theorem} \label{thm:summary}{\bf [Informal; see Section \ref{sec:theorem} for a formal description]}
For any desired accuracy level $\hat{\varepsilon}$, additive noise level  $\beta= O(\hat{\epsilon})$, and a multiplicative noise level $ \alpha$ that is a sufficiently small constant, 
there exists an algorithm that solves Problem \ref{problem1} and outputs $\hat{x}$ with high probability such that
  $F(\hat{x})-F(x^{\star}) \leq \hat{\varepsilon}$. 
The  running time of the algorithm is polynomial in $d$, $R$, $1/r'$, and $\lambda$, whenever $\alpha \leq \tilde{O}(\frac{1}{d})$  and $\beta \leq \tilde{O}(\frac{\hat{\varepsilon}}{d})$.
\end{theorem}

\noindent
When the multiplicative noise coefficient satisfies $\alpha \leq \tilde{O}(\frac{1}{d})$, Theorem \ref{thm:summary} guarantees that one can obtain an approximate minimizer $\hat{x}$ such that $F(\hat{x})-F(x^{\star}) \leq \hat{\varepsilon}$ for arbitrarily small $\hat{\varepsilon}$ in polynomial  time.  
Also note that related work \cite{applegate_kannan} for additive noise does not require a Lipschitz gradient or a bound on the diameter of $\mathcal{K}$, although they still require a bound on the range of the objective function.
The requirement that $\beta \leq \tilde{O}(\frac{\hat{\varepsilon}}{d})$ in order to get a polynomial running time can be shown to be  necessary using results from the work of \cite{blum1989training} (as done by \cite{hitting_times}).
  If the additive noise $\beta$ was required to be any lower than $\Omega(\frac{\hat{\varepsilon}}{d})$, the algorithm would take an exponentially long time to escape the local minima (Figure \ref{fig:level_sets}).
  We believe that the requirement that $\alpha \leq \tilde{O}(\frac{1}{d})$ in order to get a polynomial running time is also tight for a similar reason.  
    This is because a sub-level set $U$ of $F$ of height $\hat{\varepsilon}$, i.e., $U=\{x \in \mathcal{K}: F(x)\leq\hat{\varepsilon}\}$, will have a uniform bound on the noise of size $\sup_{x\in U}\alpha F(x) \leq \alpha \hat{\varepsilon}$ in the presence of multiplicative noise  level $\alpha$.  
  This is equivalent to having additive noise  level $\tilde{O}(\frac{\hat{\varepsilon}}{d})$, which is required for the Markov chain to quickly escape the local minima of that sub-level set.
  Establishing this formally is an interesting open problem.
  While our algorithm's running time is polynomial in various parameters, we believe that it is not tight and can be improved with a more thorough analysis of the underlying Markov chain.
The results of \cite{hitting_times} for the additive noise  is more general; their algorithm works  for a class of nonconvex functions $F$ with a certain saddle-point property.
 It would therefore be interesting to see if we can solve Problem \ref{problem1} for this class of  nonconvex functions $F$ but under the more general noise model where we have both additive and multiplicative noise.
We note the following obvious but important corollary of our main result for nonconvex functions: Suppose we are given oracle access to a nonconvex function $\hat{F}$ with a guarantee that there is a convex function $F$ such that $|\hat{F}(x) - F(x)|\leq \alpha (F(x) - F(x^{\star})) + \beta$ (as in Definition \ref{def:noise_both2}), then there is an algorithm to minimize $\hat{F}$.

\subsection{On the assumption that $F(x^\star)=0.$}\label{rem:binarysearch}

Suppose that we are given a function $\mathcal{F}$ with noisy oracle $\hat{\mathcal{F}}$ with additive and multiplicative noise level $\alpha, \beta$, but $\mathcal{F}(x^{*}) \neq 0$.  Then, if we know the value of $\mathfrak{m}=\mathcal{F}(x^{*})$, we can put this function in the form of Problem \ref{problem1} by defining a ``shifted" objective function $F(x) := \mathcal{F}(x) - \mathcal{F}(x^{*})$ and ``shifted" oracle $\hat{F}(x) := \hat{\mathcal{F}}(x) - \mathfrak{m}$.  In practice, we do not know the minimizing value $\mathfrak{m}$, but we can still obtain a noisy oracle $\hat{F}$ for $F$ by guessing a value for $\mathfrak{m}'$ and setting $\hat{F}'(x) =\hat{\mathcal{F}}(x) - \mathfrak{m}'$, although $\hat{F}'$ will have a larger additive noise level $|\mathfrak{m}' - \mathcal{F}(x^{*})| + \beta$  depending on the accuracy $|\mathfrak{m}' - \mathcal{F}(x^{*})|$ of our guess.   In practice, if we know $\beta$, then we can get around this problem by performing a binary search, by repeatedly running our algorithm using a sequence of noisy oracles $\hat{F}_1, \hat{F}_2, \ldots$ obtained with different guesses $\mathfrak{m}'_1, \mathfrak{m}'_2, \ldots$.  If we make a guess $\mathfrak{m}'_j$ and our algorithm returns a value $\hat{\mathcal{F}}(\hat{x}) \leq \mathfrak{m}'_j + \beta$, then our next guess $\mathfrak{m}'_{j+1}$ should be lower; otherwise it should be higher.  The number of times we must run our algorithm is therefore only logarithmic in the desired accuracy $\hat{\epsilon}^{-1}$.  If we do not know $\beta$, then the number of times we must run our algorithm will instead be polynomial in $\hat{\epsilon}^{-1}$, $\lambda$ and $R$.

Since our framework (Problem \ref{problem1}) assumes $F(x^{\star}) = 0$, the value of $\hat{F}(x)$ gives us a good estimate for the amount of (multiplicative) noise near a point $x$.  Therefore, $\hat{F}$ can help us choose the temperature parameter at each step in our algorithm: a larger value of $\hat{F}$ means that there may be more noise present and we require a higher temperature, while a lower value of $\hat{F}$ means that we can lower temperature.  (See Section \ref{rem:binarysearch} for how our framework can be generalized to $F(x^{*}) \neq 0$).

\begin{figure}[H]
\begin{center}
\includegraphics[trim={0 4cm 0 0}, clip, scale=0.3]{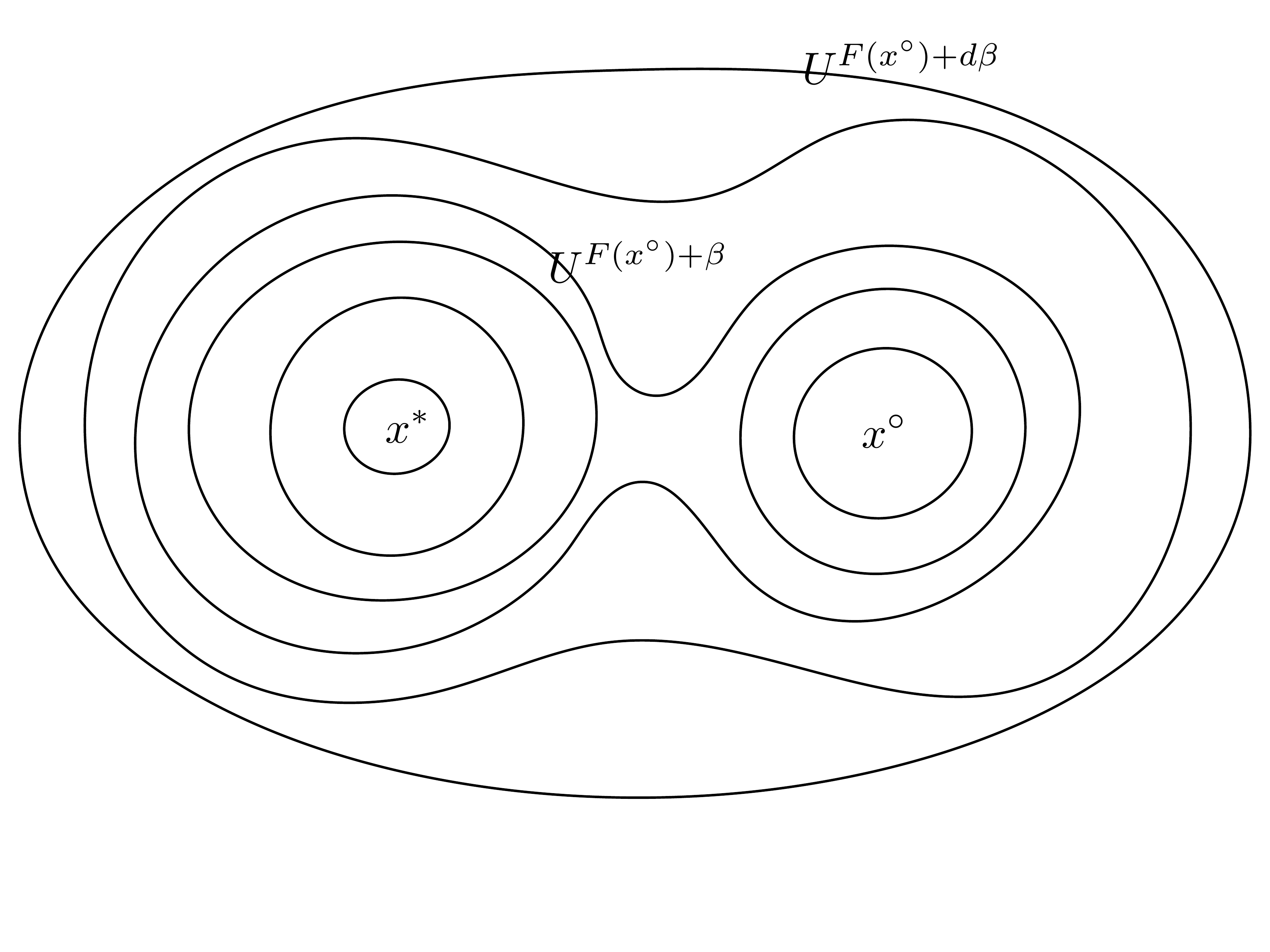}
\end{center}
\caption{To quickly escape a local minimizer $x^\circ$ of ``depth" $\beta$, a Markov chain must run at a temperature $\beta$.  At this temperature, the Markov chain will concentrate in a sub-level set of height $d\beta$. This sub-level set does not have a narrow bottleneck, so a Markov chain running at temperature $\beta$ will quickly escape the local minimum at $x^\circ$.}\label{fig:level_sets}
\end{figure}

\subsection{Short summary of techniques}

To find an approximate global minimum of the objective function $F$, we must try to find an approximate global minimum of the noisy approximation $\hat{F}$. 
 One method of optimizing a nonconvex or approximately convex function $\hat{F}$ is to generate a Markov chain with stationary distribution approximating the canonical distribution $\hat{\pi}^{(\xi)}(x):= \frac{1}{\int_{\mathcal{K}}e^{-\xi \hat{F}(y)}\mathrm{d}y}e^{-\xi \hat{F}(x)},$ where $\xi$ is thought of as an ``inverse temperature'' parameter. 
  If the ``temperature" $\xi^{-1}$ is small, then $\hat{\pi}^{(\xi)}$ concentrates near the global minima of $\hat{F}$. 
   On the other hand, to escape local minima of ``depth" $\beta>0$ in polynomial time, one requires the temperature $\xi^{-1}$ to be at least $\Omega(\beta)$ (see Figure \ref{fig:level_sets}). 
    Now consider the random variable $Z \sim N(0, \xi^{-1} I_d)$ with $\pi^{(\xi)}(x):= \frac{1}{\int_{\mathbb{R}^d}e^{-\xi \frac{1}{2} \|y\|^2}\mathrm{d}y}e^{-\xi \frac{1}{2} \|x\|^2}$.
     Then  $F(Z)$ concentrates near $d\xi^{-1}$ with high probability. 
      This suggests that for a noisy function $\hat{F}$ where we are given a bound on the additive noise  level $\beta>0$, the best we can hope to achieve in polynomial time is to find a point $\hat{x}$ such that $|F(\hat{x}) - F(x^\star)| \leq \tilde{O}(d\beta),$ since there may be sub-optimal local minima in the vicinity of $x^\star$ that have depth $O(\beta)$, requiring the temperature $\xi^{-1}$ to be at least $\Omega(\beta)$ (Figure \ref{fig:level_sets}).

As mentioned earlier, optimization of a noisy function under additive noise is studied by \cite{hitting_times}, who analyze the stochastic gradient Langevin dynamics (SGLD) Markov chain. 
 The SGLD chain approximates the Langevin diffusion, which has stationary distribution $\hat{\pi}^{(\xi)}$.  
 They show that by running SGLD at a single fixed temperature $\xi$ one can obtain an approximate global minimizer $\hat{x}$ of $F$ such that  $|F(\hat{x}) - F(x^\star)| < \tilde{O}(\hat{\varepsilon})$ with high probability with running time that is polynomial in $d$, $e^{d\beta/\hat{\varepsilon}}$, and various smoothness bounds on $F$.
   In particular, for the algorithm to get a polynomial running time in $d$ and $\beta$ one must choose $\hat{\varepsilon} = \Omega(d \beta)$.  
   Thus, the SGLD algorithm returns an approximate minimizer such that $|F(\hat{x}) - F(x^\star)| \leq \tilde{O}(d\beta)$ in polynomial time in the additive case.

More generally, if multiplicative noise is present one may have many local minima of very different sizes, so our bound on the ``depth" of the local minima is not uniform over $\mathcal{K}$.
  In this case the approach by \cite{hitting_times} of using a single fixed temperature will lead to either a very long running time or a very large error $\hat{\varepsilon}$:  If the temperature is hot enough to escape even the deepest the local minima, then the Markov chain will not concentrate near the global minimum and the error $\hat{\varepsilon}$ will be large  (Figure \ref{fig:annealing}(b)). 
   If the temperature is chosen to be too cold, then the algorithm will take a very long time to escape the deeper local minima (Figure \ref{fig:annealing}(c)).
 Instead of using a fixed temperature, we  search for the global minimum by starting the Markov chain at a high temperature and then slowly lowering the temperature at each successive step of the chain (Figure \ref{fig:annealing}(d)). This approach is referred to as ``simulated annealing" in the literature \cite{kirkpatrick1983optimization}.

The only non-asymptotic analysis we are aware of where the bound on the noise is not uniform involves a simulated annealing technique based on the hit-and-run algorithm \cite{Simulated_Annealing_Nonassymptotic}.  
  Specifically,  \cite{Simulated_Annealing_Nonassymptotic} show that if $F$ is $\mathfrak{m}$-strongly convex, then one can compute an approximate global minimizer $\hat{x}$ such that $|F(\hat{x}) - F(x^\star)| < (\frac{d}{\mathfrak{m}})^{\frac{1}{2-p}}$ with running time $\tilde{O}(d^{4.5})$, as long as $N(x) \leq c \|x\|^p$ for some $0<p<2$ and some $c>0$.
 The algorithm used by \cite{Simulated_Annealing_Nonassymptotic} runs a sequence of subroutine Markov chains.
 Each of these subroutine Markov chains is restricted to a ball $B(y_k,r_k)$ centered at the point $y_k$ returned by the subroutine chain from the last epoch.  
 Crucially, for this algorithm to work, $r_k$ must be chosen such that $B(y_k,r_k)$ contains the minimizer $x^\star$ at each epoch $k$.
Towards this end, \cite{Simulated_Annealing_Nonassymptotic} show that since the temperature is decreased at each epoch, $F(y_k)$ is much smaller than $F(y_{k-1})$ at each epoch $k$.
Since $F$ is assumed to be strongly convex,  \cite{Simulated_Annealing_Nonassymptotic} show that this decrease in $F$ implies a contraction in the distance $\|y_k-x^\star\|$ at each epoch $k$, allowing one to choose a sequence of radii $r_k$ that contract as well at each step but still have the property that $x^\star \in B(y_k,r_k)$.

One obstacle in generalizing the results by \cite{Simulated_Annealing_Nonassymptotic}  to the non-strongly convex case is that we do not have an oracle for the sub-level sets of $F$, but only for $\hat{F}$, whose sub-level sets may not even be connected.
Instead, we show that the SGLD Markov chain concentrates inside increasingly smaller {sub-level sets} of $F$ as the temperature parameter is decreased.
  To analyze the behavior of the SGLD Markov chain at each temperature, we build several new tools and use some from  the past work.
Our results make important contributions  to the growing body of work on non-asymptotic analysis of simulated annealing, Langevin dynamics and their various combinations \cite{raginsky2017non,BubeckEL15,welling2011bayesian,lee2017convergence, mangoubi_smith, mangoubi2018dimensionally}.

\subsection{Organization of the rest of the paper}
In the main body of the paper we present  a detailed but informal primer of the algorithm followed by the key steps and ideas involved in the proof of Theorem \ref{thm:summary} in Section \ref{sec:overview}.
The precise description of the algorithm and the full proofs are quite technical and are deferred to Sections \ref{sec:preliminaries} to \ref{sec:proofs}.
We present the notation and other preliminaries in Section \ref{sec:preliminaries}.
This is followed by a formal presentation of the algorithm and the statement of the main results in Sections \ref{sec:algorithm} and \ref{sec:theorem}.
Section \ref{sec:proofs} contains the detailed mathematical proof of our main theorem.

 \begin{figure}[H]
 \centering
	\begin{tabular}{cc}
		\includegraphics[scale=0.20]{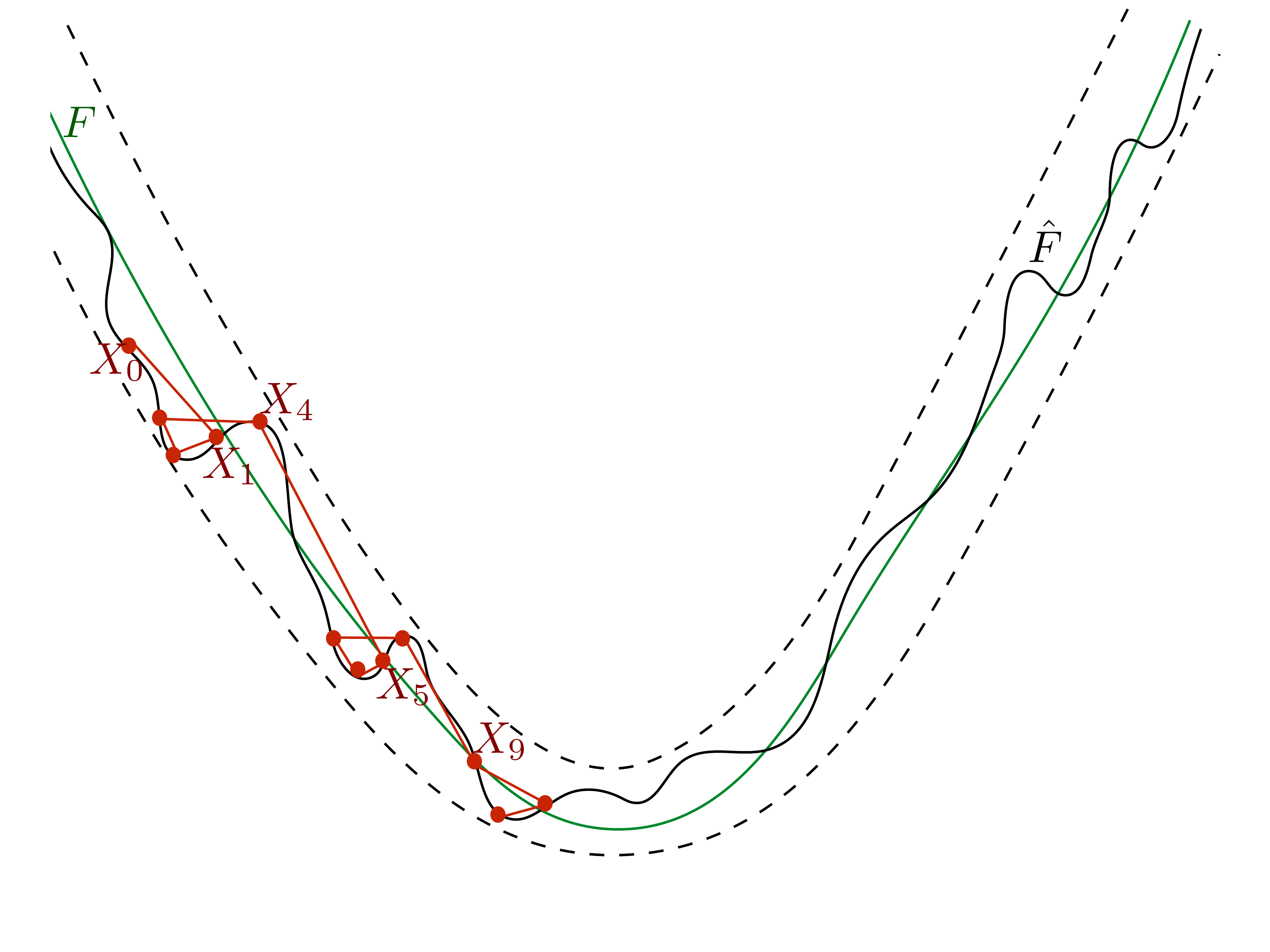}&		\includegraphics[scale=0.20]{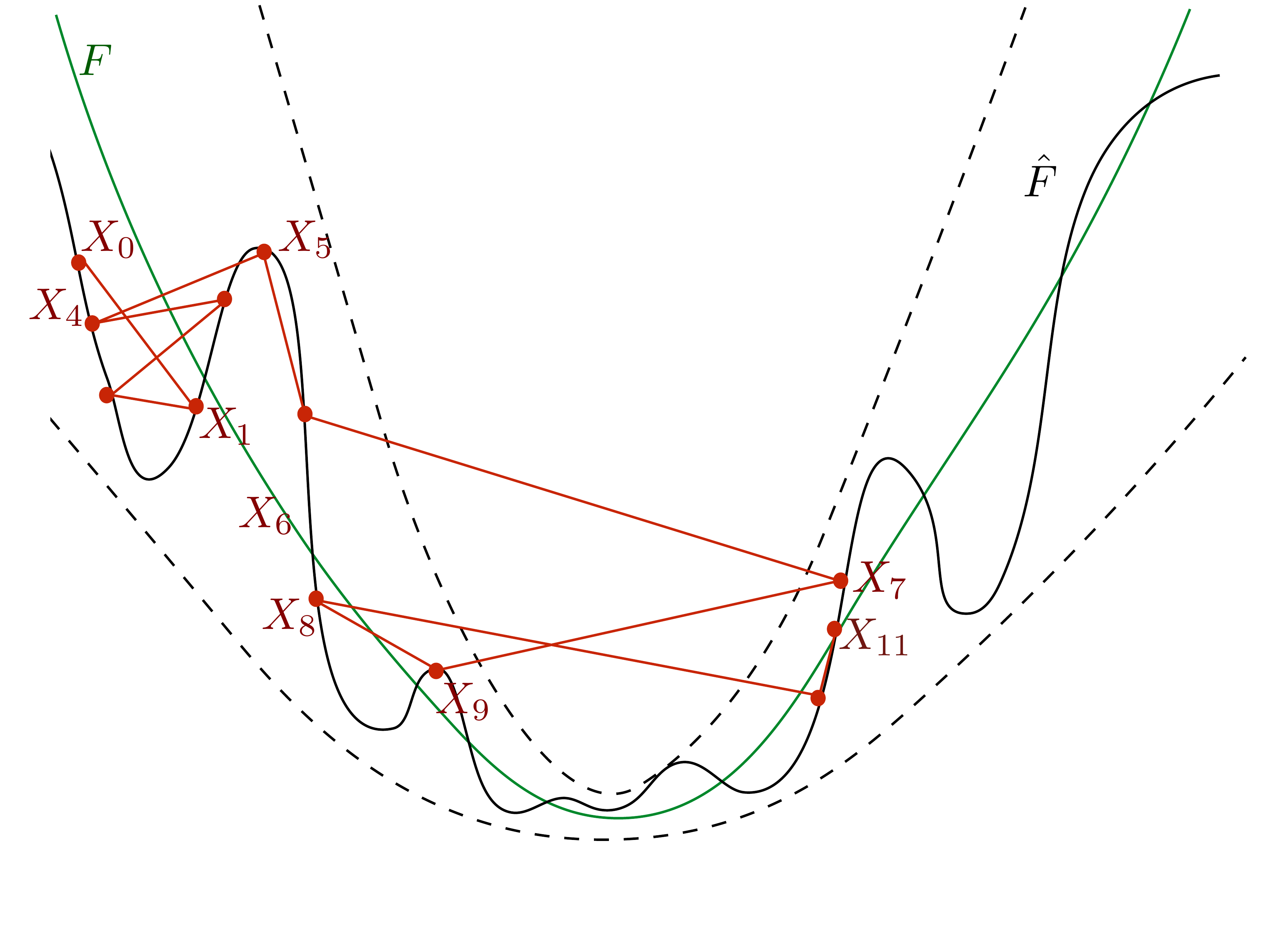}\\
		(a)&(b) 	\\ \vspace{.5in}
		\includegraphics[scale=0.20]{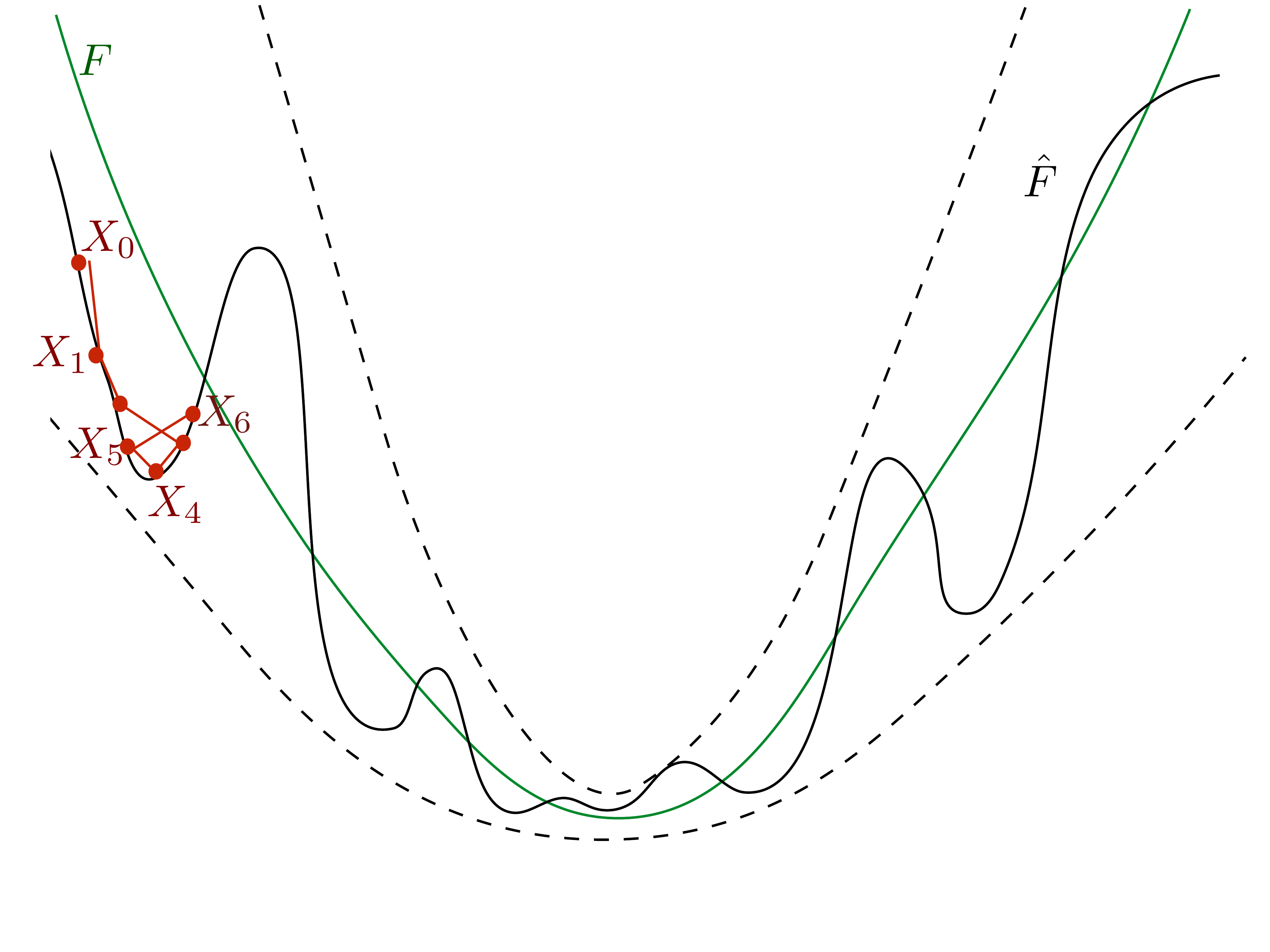}&
		\includegraphics[scale=0.20]{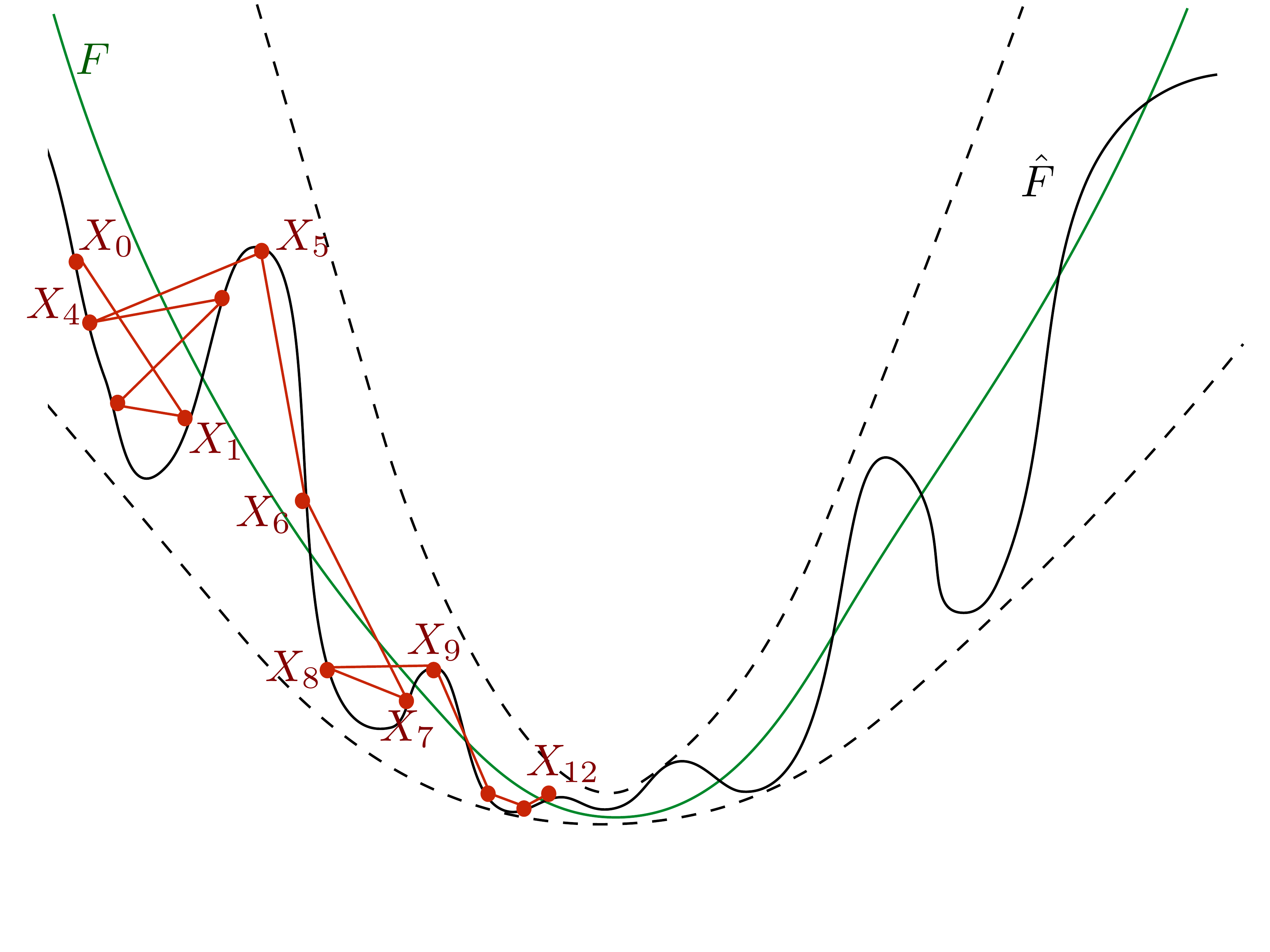}\\
		(c)&(d) 	\\
	\end{tabular}
	\caption{  (a) Optimization of a convex function $F$ (green) with noisy oracle $\hat{F}$ (black) under bounded additive noise.  
Since the gap between the noise bounds (dashed lines) is constant, the Markov chain (red) can be run at a single temperature that is both hot enough to quickly escape any local minimum but also cold enough so that the Markov chain eventually concentrates near the global minimum. (b) and (c)  {Optimization of a convex function $F$ (green) with noisy oracle $\hat{F}$ (black) when both additive and multiplicative noise are present, if we run the Markov chain at single a fixed temperature.  
If the temperature is hot enough to escape even the deepest local minima (b), then the Markov chain will not concentrate near the global minimum, leading to a large error.  
If instead the Markov chain is run at a colder temperature (c), it will take a very long time to escape the deeper local minima.} (d) Optimization of a convex function $F$ (green) with noisy oracle $\hat{F}$ (black) under both additive and multiplicative noise, when using a gradually decreasing temperature.  If multiplicative noise is present the local minima of $\hat{F}$ are very deep for large values of $F$.
  To quickly escape the deeper local minima, the Markov chain is started at a high temperature.  
  As the Markov chain concentrates in regions where $F$ is smaller, the local minima become shallower, so the temperature may be gradually decreased while still allowing the Markov chain to escape nearby local minima.  
  As the temperature is gradually decreased, the Markov chain concentrates in regions with successively smaller values of $\hat{F}$. }
  \label{fig:annealing}
	\end{figure}

\section{Overview of Our Contributions}\label{sec:overview}

\paragraph{The model and the problem.}
Let $\mathcal{K} \subseteq \mathbb{R}^d$ be the given convex body contained in a ball of radius $R>0$ and $F: \mathcal{K} \rightarrow \mathbb{R}$ the given  convex function.  We assume that $F$ has gradient bounded by some number $\lambda >0$, and that $\mathcal{K}= \mathcal{K}' + B(0,r')$ for some $r'>0$, where ``+" denotes the Minkowski sum and $\mathcal{K}'$ is a convex body.
Let  $x^\star$ be a minimizer of $F$  over $\mathcal{K}$.
Recall that our goal is to find an approximate minimizer $\hat{x}$ of $F$ on $\mathcal{K}$, such that $F(\hat{x}) - F(x^\star) \leq \hat{\varepsilon}$ for a given $\hat{\varepsilon}>0$.

We assume that we have access to a membership oracle for $\mathcal{K}$ and a noisy oracle $\hat{F}$ for $F$.
Recall that in our model of noise, {since $F(x^{*}) = 0$, we may} assume that there exist functions $\varphi:\mathcal{K} \rightarrow \mathbb{R}$ and $\psi:\mathcal{K} \rightarrow \mathbb{R}$ and numbers $\alpha, \beta \geq 0$, with $|\varphi(x)|\leq\beta$ and $|\psi(x)|\leq\alpha$ for all $x\in \mathcal{K}$, such that
\be \label{eq:noise_model}
\hat{F}(x) = F(x)(1+\psi(x)) + \varphi(x) \quad \quad \forall x \in \mathcal{K}.
\ee
We say that $\hat{F}$ has ``additive noise" $\varphi$ of level $\beta$ and ``multiplicative noise" $\psi$ of level $\alpha$.  
To simplify our analysis, we assume that $F\geq 0$ and that $F$ has minimizer $x^\star \in \mathcal{K}$ such that $F(x^\star)= 0$ (if not, we can always shift $F$ and $\hat{F}$ down by the constant $F(x^\star)$ to satisfy this assumption).  
That way, the multiplicative noise $\psi$ has the convenient property that it goes to zero as we approach $x^\star$.

We first describe our algorithm and the proof assuming that $\hat{F}$ is smooth and we have access to its gradient $\nabla \hat{F}$. 
Specifically, we assume that $\|\nabla \hat{F}\|$ is bounded above by some number $\tilde{\lambda}>0$ and that the Hessian of $\hat{F}$ has singular values bounded above by $L>0$.
This simplifies the presentation considerably and we explain how to deal with the non-smooth case at the end of this section.

\paragraph{Our algorithm.}
To find an approximate minimizer of $F$, we would like to design a Markov chain whose stationary distribution concentrates in a subset of $\mathcal{K}$ where the values of $F$ are small.
The optimal choice of parameters for this Markov chain will depend on the amount of noise present.
Since the bounds on the noise are not uniform, the choice of these parameters will depend on the current state of the chain. 
To deal with this fact, we will run a sequence of Markov chains in different epochs, where the parameters of the chain are fixed throughout each epoch.
Our algorithm runs for $k_{\textrm{max}}$ epochs, with each epoch indexed by $k$.

 In epoch $k$, we run a separate Markov chain $\{X_i^{(k)}\}_{i=1}^{{i_{\textrm{max}}}}$ over $\mathcal{K}$  
 for the same number of iterations $i_{\textrm{max}}$.
 Each such Markov chain has parameters $\xi_k$ and $\eta_k$ that depend on $k$.  
 We think of $\xi_k^{-1}$ as the ``temperature" of the Markov chain and $\eta_k$ as the step size.  
 At the beginning of each epoch,  we decrease the temperature and step size, and keep them fixed throughout the epoch.
We explain quantitatively how we set the temperature a bit later. 
Each Markov chain also has an initial point $X_0^{(k)} \in \mathcal{K}$.  
This initial point is chosen from the uniform distribution on a small ball  centered at the point in the Markov chain of the previous epoch $(k-1)$ with the smallest value of $\hat{F}$.
In the final epoch, the algorithm outputs a solution $\hat{x}$, where $\hat{x}$ is chosen to be the point in the Markov chain of the final epoch with the smallest value of $\hat{F}$.

\paragraph{Description of the Markov chain in a single epoch.}
We now describe how the Markov chain,  at the point $X_i^{(k)}$ in the $k$-th epoch, chooses the next point $X_{i+1}^{(k)}$. 
 First, we compute the gradient $\nabla \hat{F}(X_i^{(k)})$.  
 Then we compute a ``proposal'' $X'_{i+1}$ for the next point as follows
\be \label{update_rule}
X'_{i+1} = X^{(k)}_i - \eta_k \nabla \hat{F}(X^{(k)}_i) + \sqrt{\frac{2\eta_k}{\xi_k}} P_i,
\ee
where $P_i$ is sampled from $N(0,I_d)$.
 If $X'_{i+1}$ is inside the domain $\mathcal{K}$, then we accept the proposal and set $X^{(k)}_{i+1} = X'_{i+1}$; otherwise we reject the proposal and we set $X^{(k)}_{i+1} = X^{(k)}_i$ -- which is the old point.  
 The update rule in Equation \eqref{update_rule} is called the Langevin dynamics.  
 This is a version of gradient descent injected with a random term.  
 The amount of randomness is controlled by the temperature $\xi_k^{-1}$ and the step size $\eta_k$.
 This randomness allows the Markov chain to escape local minima when $\hat{F}$ is not convex. 
  Although the stationary distribution of this Markov chain is not known exactly, roughly speaking it is approximately proportional\footnote{{By this we actually mean that the Markov chain is $\epsilon'$-close to another Markov chain $Z$ with stationary distribution $\propto e^{-\xi_k \hat{F}}$ (see Definition \ref{def:epsilon_close} and  Theorem \ref{thm:error}). Specifically, a Markov chain $X$ on a space $S$ with transition kernel $Q_X$ is said to be $\epsilon'$-close to a Markov chain $Z$ w.r.t. a set $U$ if $Q_Z(x,A) \leq  Q_X(x,A) \leq (1+\epsilon)Q_Z(x,A)$ for all $x \in \mathcal{K} \backslash U$ and $A \subseteq \mathcal{K} \backslash\{x\}$ \cite{hitting_times}}} to  $e^{-\xi_k \hat{F}}$. 
 This completes the description of our algorithm in the smooth case and we now turn to explaining the steps involved in bounding its running time for a given bound on the error $\hat{\varepsilon}$.

\paragraph{Steps in bounding the running time.}
In every epoch, the algorithm makes multiplicative progress so that the smallest value of $F$ achieved by the Markov chain decreases by a factor of $\nicefrac{1}{10}$.  
To achieve an error $\hat{\varepsilon}$, our algorithm therefore requires $k_{\mathrm{max}} = O(\log \frac{M}{\hat{\varepsilon}})$ epochs, where $M$ is the maximum value of $\hat{F}$ on $\mathcal{K}$ ($M \leq \lambda R$).  
The running time of our algorithm is given by the number of epochs $k_{\mathrm{max}}$ multiplied by the number of steps $i_{\mathrm{max}}$ taken by the Markov chain within each epoch.  
For simplicity, we will run the Markov chain at each epoch for the same number of steps $i_{\mathrm{max}}$.
For the value of $F$ to decrease by a factor of $\nicefrac{1}{10}$ in each epoch, we must set the number of steps $i_{\mathrm{max}}$ taken by the Markov chain during each epoch to be no less than the hitting time of the Markov chain for epoch $k$ to a sub-level set $U_k\subseteq \mathcal{K}$ of $F$, where the ``height" of $U_k$ is one-tenth the value of $F$ at the initial point in this Markov chain.  
By the height of a sub-level set, we mean the largest value of $F$ achieved at any point on that sub-level set, that is the sub-level set $\{y\in \mathcal{K}: F(y) \leq h\}$ has height $h$.  
Thus, bounding the hitting time will allow us to bound the number of steps $i_{\mathrm{max}}$ for which we must run each Markov chain.  Specifically, we should choose $i_{\mathrm{max}}$ to be no less than the greatest hitting time in any of the epochs with high probability. %

 This approach was used in the simpler setting of additive noise and a non-iterative way by \cite{hitting_times}.  
 Thus, the running time is roughly the product of the number of epochs and the hitting time to the sub-level set $U_k$,
  and having determined the number of epochs required for a given accuracy, we proceed to bounding the hitting time.

\paragraph{Bounding the hitting time and the Cheeger constant.}

To bound the hitting time of the Markov chain in a single epoch, we use the strategy of \cite{hitting_times}, who bound the hitting time of the Langevin dynamics Markov chain in terms of the Cheeger constant.  
Since the Markov chain has approximate stationary measure induced by $e^{-\xi_k \hat{F}}$, we consider the Cheeger constant with respect to this measure, defined as follows:

Given a probability measure $\mu$ on some domain, we consider the ratio of the measure of the boundary of an arbitrary subset $A$ of the domain to the measure of $A$ itself.  
The Cheeger constant of a set $V$ is the infimum of these ratios over all subsets $A \subseteq V$ (see Definition \ref{def:Cheeger} in  Section \ref{sec:conductance} for a formal definition).  
  We  use some of the results of \cite{hitting_times} to show a bound on the hitting time to the sub-level set $U_k$ contained in a larger sub-level set $U'_k$ in terms of the Cheeger constant $\hat{\mathcal{C}}_k$, with respect to the measure induced by $e^{-\xi_k\hat{F}(x)}$ on $U'_k$. 
   Specifically, we  set $U'_k$ to be the sub-level set of height $\hat{F}(X^{(k)}_0) + \xi_k^{-1}d$ and $U_k$ to be the sub-level set of height $\frac{1}{10}F(X^{(k)}_0)$ and show that for a step size 
   $$\eta_k = \frac{(\hat{\mathcal{C}}_k)^2}{d^3((\xi_k \tilde{\lambda})^2 + \xi_k L)^2},$$
    the hitting time to $U_k$ is bounded by $\frac{R \tilde{\lambda} \xi_k + d}{\sqrt{\frac{\eta_k}{d}} \hat{\mathcal{C}}_k}$; see Section \ref{sec:result_smooth}.
      Thus, to complete our bound on the hitting time we need to  bound  the corresponding Cheeger constants.

\paragraph{Bounding the Cheeger constant.}  We would like to bound the Cheeger constant of the measure induced by $e^{-\xi_k\hat{F}(x)}$.  
However, $\hat{F}$ is not convex, so we cannot directly apply the usual approach of  \cite{lovasz1993random} for convex functions. 
Instead, we first apply their result to bound the Cheeger constant of the convex function $F$. 
We then bound the Cheeger constant of the nonconvex function $\hat{F}$ in terms of the Cheeger constant of the convex function $F$, using a very useful stability property satisfied by the Cheeger constant.

Roughly speaking, we show that the Cheeger constant of $U'_k \backslash U_k$ is bounded below by $\nicefrac{1}{R}$ (where $R$ is the radius of the bounding ball for $\mathcal{K}$) as long as the inverse temperature satisfies 
\be 
\xi_k \geq \frac{d}{\nicefrac{1}{10}F(X^{(k)}_0)}
\ee (see Lemma \ref{lemma:cheeger2}).
However, the difficulty is that since $U'_k$ may have sharp corners, the volume of $U_k$ might be so small that $U_k$ would have much smaller measure than $U'_k \backslash U_k$, leading to a very small Cheeger constant.
To get around this problem, we instead consider a slightly ``rounded" version of $\mathcal{K}$, where we take $\mathcal{K}$ to be the Minkowski sum of another convex body with a ball of very small radius $r'$.
The roundness allows us to show that $U_k$ contains a ball of even smaller radius $\hat{r}$ such that the measure is much larger on this ball than at any point in  $U'_k\backslash U_k$. 
This in turn allows us to apply the results of \cite{lovasz1993random} to show that the Cheeger constant is bounded below by $\nicefrac{1}{R}$ (see Lemma \ref{lemma:cheeger2}).  
Note our Cheeger bound is more general   (for convex functions) than that obtained by \cite{hitting_times}, where the constraint set is assumed to be a ball.

 The Cheeger constant has the following useful stability property that allows us to bound the Cheeger constant of the nonconvex $\hat{F}$ with respect to the convex $F$: if $|\hat{F}(x) - F(x)| \leq N_k$ for all $x\in U'_k$, then the Cheeger constant for the measures proportional to $e^{-\xi_k F}$ and $e^{-\xi_k\hat{F}}$ differ by a factor of at most $e^{-2\xi_k N_k}$.
For our choice of $U'_k$, we have  
\be
N_k \approx \alpha[F(X^{(k)}_0) + \xi_k^{-1}d]+\beta.
\ee
We can then use the stability property to show that the Cheeger constants of $\hat{F}$ and $F$ differ by a factor of at most $e^{-2\xi_k N_k}$, allowing us to get a large bound for the Cheeger constant of $\hat{F}$ in terms of our bound for the Cheeger constant of $F$ as long as the bound on the noise $N_k$ on $U'_k$ is not too large, namely we get that the Cheeger constant is bounded below by 
\be
\frac{1}{R} e^{-2\xi_k N_k} \approx \frac{1}{R} \exp\left(-\alpha d - \frac{d}{F(X^{(k)}_0)}\beta \right)
\ee
 if we choose $\xi_k = \frac{d}{\frac{1}{10}F(X^{(k)}_0)}$.

At this point we mention the key difference between the approach of  \cite{hitting_times}  and ours in bounding the hitting time.
As \cite{hitting_times} assume a uniform bound on the noise they only consider the Cheeger constant of $\mathcal{K}\backslash U_k$, where $\mathcal{K}$ is the entire constraint set and is assumed to be a ball.  
Since the noise in our model depends on the ``height'' of the level sets, we  instead need to bound the Cheeger constant of $U'_k\backslash U_k$,  where $U'_k$ is the level set of height $\hat{F}(X^{(k)}_0) + \xi_k^{-1}d$ and $U_k$ is the level set of height $\frac{1}{10}F(X^{(k)}_0)$.

In order to complete our bound for the Cheeger constant of $\hat{F}$, we still need to verify that we can choose a temperature such that the Cheeger constant of $F$ is large and the Cheeger constants of $F$ and $\hat{F}$ are close at this same temperature.

\paragraph{Requirements on the temperature to bound the Cheeger constant.}
To get a large bound for the Cheeger constant of $\hat{F}$, we  need to use a temperature $\xi_k^{-1}$ such that the following competing requirements are satisfied:
\begin{enumerate}
\item   We want the Cheeger constant of the convex objective function $F$ on $U'_k\backslash U_k$ to be bounded below by $\nicefrac{1}{R}$.  We can show such a bound on the Cheeger constant if the temperature is {\em low} enough, in particular a temperature of $\xi_k^{-1} \approx \frac{\frac{1}{10}F(X^{(k)}_0)}{d}$ suffices.
\item We need the Markov chain to stay inside a level set on which the upper bound $N_k$ on the noise is at most $\alpha[F(X^{(k)}_0) + \xi_k^{-1}d]+\beta$, to show that the Cheeger constants of $F$ and $\hat{F}$ are close.  That is, we need to show that the ratio $e^{-2\xi_k N_k}$ of the Cheeger constants of $F$ and $\hat{F}$ is not too small, roughly $e^{-2\xi_k N_k} \geq \exp\left(-\alpha d - \frac{d}{F(X^{(k)}_0)}\beta\right)$.  This again requires the temperature to be {\em low} enough, with $\xi_k^{-1} \approx \frac{\frac{1}{10}F(X^{(k)}_0)}{d}$ sufficing.
\item To show that the ratio of the Cheeger constants roughly satisfies $e^{-2\xi_k N_k} \geq \exp\left(-\alpha d - \frac{d}{F(X^{(k)}_0)}\beta\right)$, we also need the temperature to be {\em high} enough.  Specifically, a temperature of $\xi_k^{-1} \approx \frac{\frac{1}{10}F(X^{(k)}_0)}{d}$ suffices for this requirement as well.
\end{enumerate}
At some epoch $k$, the value of $F$ becomes too low for all three of these requirements on the temperature to be satisfied simultaneously.  At this point the Cheeger constant and hitting time to $U_k$ become very large no matter what temperature we use, so that the minimum value of $F$ obtained by the Markov chain no longer decreases by a large factor in $i_{\mathrm{max}}$ steps.

\vspace{-2mm}
\paragraph{Quantitative error and running time bounds.} 
  We  now give a more quantitative analysis to determine at what point $F$ stops decreasing. 
   The value of $F$ at this point  determines the error $\hat{\epsilon}$  of the solution returned by our algorithm.
Towards this end, we set the inverse temperature to be $\xi_k = \frac{d}{\frac{1}{10}F(X^{(k)}_0)}$ and  check to what extent all 3 requirements above are satisfied.
\begin{enumerate}
\item We start by showing that if the temperature roughly satisfies $\xi_k^{-1} \leq \frac{\frac{1}{10}F(X^{(k)}_0)}{d}$ then the Cheeger constant for $F$ on $U'_k\backslash U_k$ is bounded below by $\nicefrac{1}{R}$ (see Lemma \ref{lemma:cheeger2}).
\item   We then show that at each epoch the Markov chain remains with high probability in the level set $U'_k$ of height $\hat{F}(X^{(k)}_0) + \xi_k^{-1}d$  (Lemma \ref{lemma:drift2}).  The fact that the noise satisfies $|\hat{F}(x)-F(x)| \leq \alpha F(x) + \beta$ (note that we assume $F \geq 0$),  implies that the noise is roughly bounded above by $N_k = \alpha[F(X^{(k)}_0) + \xi_k^{-1}d]+\beta$ on this level set.
\item Since we chose the temperature to be $\xi_k^{-1} = \frac{\frac{1}{10}F(X^{(k)}_0)}{d}$, we have that
 \be
 \xi_k N_k
 \approx \alpha d + \frac{d}{F(X^{(k)}_0)}\beta.
 \ee
\end{enumerate}
Combing these three facts we get that the Cheeger constant is bounded below by $\frac{1}{R} e^{-2\xi_k N_k} \approx \frac{1}{R} \exp\left(-\alpha d - \frac{d}{F(X^{(k)}_0)}\beta\right)$.
If we run the algorithm for enough epochs to reach $F(X^{(k)}_0) \leq \hat{\varepsilon}$ for any desired error $\hat{\varepsilon}>0$, the Cheeger constant will be roughly bounded below by $\frac{1}{R} \exp(-\alpha d - \frac{d}{\hat{\varepsilon}}\beta)$.

Recall that the hitting time is bounded by $\frac{R \tilde{\lambda} \xi_k + d}{\sqrt{\frac{\eta_k}{d}} \hat{\mathcal{C}}_k}$, for stepsize $\eta_k \approx \frac{(\hat{\mathcal{C}}_k)^2}{R d^3((\xi_k \tilde{\lambda})^2 + \xi_k L)^2}$. 
 Choosing $i_{\textrm{max}}$ to be equal to our bound on the hitting time, and recalling that  $k_{\mathrm{max}} = \tilde{O}(1)$, we get a running time of roughly 
$$\tilde{O}\left(R^{\frac{3}{2}}\left[d^5
\frac{\tilde{\lambda}^3}{\hat{\varepsilon}^3} + d^{\frac{5}{2}}
\frac{L}{\hat{\varepsilon}}\right]\exp(c[\alpha d + \frac{d}{\hat{\varepsilon}}\beta])\right),$$
 for some $c = \tilde{O}(1).$

Therefore, for our choice of inverse temperature $\xi_k = \frac{d}{\frac{1}{10}F(X^{(k)}_0)}$,  the running time is polynomial in $d, R, \lambda$ and $\tilde{\lambda}$ whenever the multiplicative noise level satisfies $\alpha \leq \tilde{O}(\frac{1}{d})$  and the additive noise level satisfies $\beta \leq \tilde{O}(\frac{\hat{\varepsilon}}{d})$.  
As discussed in the introduction,  the requirements that $\alpha \leq \tilde{O}(\frac{1}{d})$  and $\beta \leq \tilde{O}(\frac{\hat{\varepsilon}}{d})$ are not an artefact of the analysis or algorithm and are in fact tight. 

\paragraph{Drift bounds and initialization.}
So far we have been implicitly assuming that the Markov chain does not leave $U'_k$, so that we could analyze the Markov chain using the Cheeger constant on $U'_k$.  
We  now show that this assumption is indeed true with high probability.
This is important to verify, since there are examples of Markov chains where the Markov chain may have a high probability of escaping a level set $U'_k$, even if this level set contains most of the stationary measure, provided that the Markov chain is started far from the stationary distribution.

To get around this problem, at each epoch we choose the initial point $X_0^{(k)}$ from the uniform distribution on a  ball of radius $r$ centered at the point in the Markov chain of the previous epoch $k-1$ with the smallest value of $\hat{F}$.  
We then show that if the Markov chain is initialized in this small ball, it has a low probability of leaving the level set $U'_k$ (see Propositions \ref{thm:drift}, \ref{lemma:drift2} and Lemma  \ref{lemma:drift2}).

Our method of initialization is another crucial difference between our algorithm and the algorithm by \cite{hitting_times} and \cite{Simulated_Annealing_Nonassymptotic}, since it allows us to effectively restrict the Markov chain to a sub-level set of the objective function $F$, which we do not have direct oracle access to, rather than restricting the Markov chain to a large ball as by \cite{Simulated_Annealing_Nonassymptotic} or the entire constraint set $\mathcal{K}$ as by \cite{hitting_times} for which we have a membership oracle.  
This in turn allows us to get a tighter bound on the multiplicative noise than would otherwise be possible, since the amount of multiplicative noise depends, by definition, on the sub-level set.

We still need to show that the chain $X^{(k)}$ does not leave the set $U'_k$ with high probability. 
To bound the probability that  $X^{(k)}$ leaves $U'_k$, we would like to use the fact that most of its stationary distribution is concentrated in $U'_k$. 
However, the problem is that we do not know the stationary distribution of $X^{(k)}$.
To get around this, we consider a related Markov chain $Y^{(k)}$ with known stationary distribution. 
The chain $Y^{(k)}$ evolves according to the same update rules as $X^{(k)}$, using the same sequence of Gaussian random vectors $P_1,P_2,\ldots$ and the same starting point, except that it performs a Metropolis ``accept-reject'' step that causes its stationary distribution to be proportional to $e^{-\xi_k \hat{F}}$. 
The fact that we know the stationary distribution of $Y^{(k)}$ is key to showing that $Y^{(k)}$ stays in the subset $U'_k$ with high probability (see Proposition \ref{lemma:drift}).  
We then argue that $Y^{(k)} = X^{(k)}$ with high probability, implying that $X^{(k)}$ also stays inside  $U'_k$  with high probability (see Lemma \ref{lemma:drift2}).

\paragraph{Another coupled toy chain.}
So far we have shown that the Markov chain $X^{(k)}$ stays inside the set $U'_k$ \emph{with high probability}.  
However, to use the stability property to bound the hitting time of the Markov chain $X^{(k)}$ to the set $U_k$, we actually want  $X^{(k)}$ to be \emph{restricted} to the set $U'_k$ where the noise is not too large.  
In reality, however, the domain of $X^{(k)}$ is all of $\mathcal{K}$, so we cannot directly bound the hitting time of $X^{(k)}$ with the Cheeger constant of $U'_k\backslash U_k$.
 Instead, we consider a Markov chain $\hat{X}^{(k)}$ that evolves according to the same rules as $X^{(k)}$, except that it rejects any proposal outside of $U'_k$.  Since  $\hat{X}^{(k)}$ has domain $U'_k$, we can use our bound on the Cheeger constant of $U'_k\backslash U_k$ to obtain a bound on the hitting time of $\hat{X}^{(k)}$.  
 Then, we argue that since $X^{(k)}$ stays in $U'_k$ with high probability, and $\hat{X}^{(k)}$ and $X^{(k)}$ evolve according to the same update rules as long as $X^{(k)}$ stays inside $U'_k$, $\hat{X}^{(k)} = X^{(k)}$ with high probability as well, implying a hitting time bound for $X^{(k)}$.

\paragraph{Rounding the sub-level sets.}
We must also show a bound on the roundness of the sets $U'_k$, to avoid the possibillty of the Markov chain getting stuck in ``corners". 
  \cite{hitting_times} take this as an assumption about the constraint set.  
 However, since we must consider the Cheeger constant on sub-level sets $U'_k$ rather than just on the entire constraint set, we must make sure that these sub-level sets are ``round enough".  
 Towards this end we consider ``rounded" sub-level sets where we take the Minkowski sum of $U'_k$ with a ball of a small radius $r'$.  
 We then apply the Hanson-Wright inequality to show that any Gaussian random variable with center inside this rounded sub-level set and small enough covariance  remains inside the rounded sub-level set with high probability (see Lemma \ref{lemma:contraint_round}).

\vspace{-2mm}
\paragraph{Smoothing a non-differentiable noisy oracle.}  Finally, so far we  have considered the special case where $\hat{F}$ is smooth.
  However, $\hat{F}$ may not be smooth or may not even be differentiable, so we may not have access to a well-behaved gradient which we need to compute the Langevin dynamics Markov chain (Equation \ref{update_rule}).  
  To get around this problem, we follow the approach of \cite{duchi2015optimal} and \cite{hitting_times}. 
   We define a smoothed function
$$
\tilde{f}_\sigma(x) := \mathbb{E}_Z[\hat{F}(x+Z)]
$$
where $Z \sim N(0, \sigma I_d)$ and $\sigma>0$ is a parameter we must fix.  
The smoothness of  $\tilde{f}_\sigma$ comes from the fact that $\tilde{f}_\sigma$ is a convolution of $\hat{F}$ with a Gaussian distribution.

When choosing $\sigma$, we want $\sigma$ to be small enough so that we get a good bound on the noise $|\tilde{f}_\sigma(x)-F(x)|$. 
 Specifically, we need 
 $$\sigma = \tilde{O}\left(\min\left(\frac{r}{\sqrt{d}}, \frac{\beta +\hat{\varepsilon}\alpha +\hat{\varepsilon}/d}{\lambda\sqrt{d}}\right)\right),$$ where $\lambda$ is a bound on $\|\nabla F\|$.  
 On the other hand, we also want $\sigma$ not to be too small so that we get a good bound on the smoothness of $\tilde{f}_\sigma$.

Further, so far we have also assumed that we have access to the full gradient of $\hat{F}$, but in general $\hat{F}$ may not even have a gradient.  
Instead, we would like to use the gradient of $\tilde{f}_\sigma$ to compute the proposal for the Langevin dynamics Markov chain (Equation \eqref{update_rule}).  
However, computing the full gradient of $\tilde{f}_\sigma$ can be expensive, since we do not even have direct oracle access to $\tilde{f}_{\sigma}$.  
Instead, we compute a projection $g(x)$ of $\nabla \tilde{f}_\sigma$, where
\be
g(x) = \frac{Z}{\sigma^2}(\hat{F}(x+Z)-\hat{F}(x))
\ee
Since $g$ has the property that $\mathbb{E}[g(x)] = \nabla \tilde{f}_\sigma$, $g$ is called a ``stochastic gradient" of $\tilde{f}_\sigma$. 
 We use this stochastic gradient $g$ in place of the full gradient of $\hat{F}$ when computing the proposal for the Langevin dynamics Markov chain (Equation \ref{update_rule}).  
 This gives rise to the following Markov chain proposal, also known as stochastic gradient Langevin dynamics (SGLD):
\be
X'_{i+1} = X^{(k)}_i - \eta_k g(X^{(k)}_i) + \sqrt{\frac{2\eta_k}{\xi_k}} P_i.
\ee
To bound the running time of SGLD, we will need a bound on the magnitude of the gradient of $\tilde{f}_\sigma$ (see Lemma \ref{lemma:smooth_gradient}), bounds on the Hessian and tails of $\tilde{f}_\sigma$, which we obtain from \cite{hitting_times} (see Lemma \ref{lemma:Hessian}), and bounds on the noise of the smoothed function, $|\tilde{f}_\sigma - F(x)|$ (see Lemma \ref{lemma:noise_smooth}). 

Although in this technical overview we largely showed running time and error bounds assuming access to a full gradient, in reality we prove Theorem \ref{thm:summary} for the more general stochastic gradient Langevin dynamics algorithm, where we only assume access to a stochastic gradient of a smooth function.  
Therefore, the bounds on the noise and smoothness of $\tilde{f}_\sigma$ allow us to extend the error and polynomial running time bounds shown in this overview to the more general case where $\hat{F}$ may not be differentiable.

\section{Preliminaries}\label{sec:preliminaries}

In this section we go over notation and assumptions that we use to state our algorithm and  prove our main result.  
We start by giving assumptions we make about the convex objective function $F$.  
We then explain how to obtain an oracle for the gradient of the smoothed function $\tilde{f}_\sigma$ if we only have access to the non-smooth oracle $\hat{F}$.

\subsection{Notation}\label{sec:notation}
In this section we define the notation we use to prove our main result.
 For any set $S\subseteq \mathbb{R}^d$ and $t\geq 0$ define $S_t := S +B(0,t)$ where ``+" denotes the Minkowski sum.
 We denote the $\ell_2$-norm by $\|\cdot \|$, and the $d\times d$ identity matrix by $I_d$.
 We denote by $ \|\cdot \|_{\mathrm{op}}$ the operator norm of a matrix, that is, its largest singular value.
 We define $B(a,t)$ to be the closed Euclidean ball with center $a$ and radius $t$.
Denote the multivariate normal distribution with mean $m$ and covariance matrix $\Sigma$ by $N(m,\Sigma)$.
Let $x^\star$ denote a minimizer of $F$ on $\mathcal{K}$.

\subsection{Assumptions on the convex objective function and the constraint set}
We make the following assumptions about the convex objective function $F$ and  $\mathcal{K}$:
\begin{itemize}

\item  $\mathcal{K}$ is contained in a ball, with $\mathcal{K} \subseteq B(c,R)$ for some $c \in \mathbb{R}^d$.

\item $F(x^\star) = 0$.\footnote{If $F(x^\star)$ is nonzero, we can define a new objective function $F'(x) = F(x)-F(x^\star)$ and a new noisy function $\tilde{f}'(x) =\tilde{f}(x)-F(x^\star)$.
The noise $N'(x)= \tilde{f}'(x) -F'(x)$ can then be modeled as having additive noise  of level $\beta' = \beta + \alpha F(x^\star)$ and multiplicative noise of level $\alpha'= \alpha$, if $N(x) = \tilde{f}(x)- F(x)$ has additive noise of level $\beta$ and multiplicative noise of level $\alpha$.}

\item There exists  an $r' > 0$  and a convex body $\mathcal{K}'$ such that   $\mathcal{K}=  \mathcal{K}' +B(0,r')$ for some convex body.
(This assumption is necessary to ensure that our convex body does not have ``pointy" edges, so that the Markov chain does not get stuck for a long time in a corner.)

\item $F$ is  convex over  $\mathcal{K}_{\mathsf{r}}$ for some $\mathsf{r}>0$.
  
\item $\| \nabla F(x) \| \leq \lambda $ for all $x \in \mathcal{K}_{\mathsf{r}}$, where $\lambda>0$.

\end{itemize}

\subsection{A smoothed oracle from a non-smooth one} \label{sec:smoothed_from_nonsmooth}

In this section we show how to obtain a smooth noisy oracle for $F$ if one only has access to a non-smooth and possibly non-continuous noisy oracle $\hat{F}$.
Our goal is to find an approximate minimum for $F$ on the constraint set $\mathcal{K}$. (We consider the thickened set $\mathcal{K}_{\mathsf{r}}$ only to help us compute a smooth oracle for $F$ on $\mathcal{K}$).
We assume that we have access to a noisy function $\hat{F}$ of the form
\be \label{eq:assumption}
\hat{F}(x) = F(x)(1+\psi(x)) + \varphi(x),
\ee
where $|\psi(x)|< \alpha$, and $|\varphi(x)|< \beta$ for every $x \in \mathcal{K}_{\mathsf{r}}$, for some $ \alpha, \beta \geq 0$.
We extend $\hat{F}$ to values outside $\mathcal{K}_{\mathsf{r}}$ by setting $\hat{F}(x) = 0$ for all $x \notin \mathcal{K}_{\mathsf{r}}$.
Since $\hat{F}$ need not be smooth, as in \cite{duchi2015optimal} and \cite{hitting_times} we will instead optimize the following smoothed function
\be \label{eq:smoother}
\tilde{f}_\sigma(x) := \mathbb{E}_Z[\hat{F}(x+Z)]
\ee
where $Z \sim \mathcal{N}(0,\sigma I_d)$, for some $\sigma>0$.  The parameter $\sigma$ determines the smoothness of $\tilde{f}_\sigma$; a larger value of $\sigma$ will mean that $\tilde{f}_\sigma$ will be smoother. 
The gradient of $\tilde{f}_\sigma(x)$ can be computed using a stochastic gradient $g(x)$, where
\be
g(x) \equiv g_Z(x) := \frac{1}{\sigma^2}Z \left(\hat{F}(x+Z)- \hat{F}(x)\right), \quad \quad \nabla \tilde{f}_\sigma(x) = \mathbb{E}_Z[g(x)].
\ee

\section{Our Contribution}

\subsection{Our Algorithm} \label{sec:algorithm}

In this section we state our simulated annealing algorithm (Algorithm \ref{alg:annealing}) that we use to obtain a solution to Problem \ref{problem1}.  
At each epoch, our algorithm uses the SGLD Markov chain as a subroutine, which we describe first in Algorithm \ref{alg:SGLD}.  
The SGLD Markov chain we describe here is the same algorithm used in \cite{hitting_times}, except that we allow the user to specify the initial point.
\begin{algorithm}[H]
\caption{Stochastic gradient Langevin dynamics (SGLD) \label{alg:SGLD}} 
\flushleft
\textbf{input:} Convex constraint set $\hat{\mathcal{K}} \subseteq \mathbb{R}^d$,  inverse temperature $\xi>0$, step size $\eta > 0$, parameters $i_{\max} \in \mathbb{N}$ and $D>0$, and a stochastic gradient oracle $g$ for some $\tilde{f}:{{\mathcal{K}}} \rightarrow \mathbb{R}$. \\
 \textbf{input:}   Initial point $X_0 \in \hat{\mathcal{K}}$.\\
\begin{algorithmic}[1]
\For{$i=0$ to $i_{\mathrm{max}}$}
\\ Sample $P_i \sim N(0,I_d)$.
\\ Set $X'_{i+1} = X_i - \eta g(X_i) + \sqrt{\frac{2\eta}{\xi}} P_i$.
\\ Set $X_{i+1} = X'_{i+1}$ if $X'_{i+1} \in \hat{\mathcal{K}} \cap B(X_i, D)$.  Otherwise, set $X_{i+1} = X_i$.
\EndFor
\end{algorithmic}
\textbf{output:} $X_{i^\star}$, where $i^\star := \mathrm{argmin}_i \{\hat{F}(X_i)\}$
\end{algorithm}

\noindent
Using Algorithm \ref{alg:SGLD} as a subroutine, we define the following simulated annealing algorithm:

\begin{algorithm}[H]
\caption{Simulated annealing SGLD \label{alg:annealing}} 
\flushleft
\textbf{input:} Convex constraint set $\hat{\mathcal{K}} \subseteq \mathbb{R}^d$, initial point $x_0 \in \hat{\mathcal{K}}$,  inverse temperatures $\xi_0,\xi_1,\ldots, \xi_{k_{\mathrm{max}}}$, step sizes $\eta_0, \eta_1,\ldots, \eta_{k_{\mathrm{max}}}$, parameters $k_{\max}, i_{\max} \in \mathbb{N}$, $D>0$ and $r>0$, and a stochastic gradient oracle $g$ for some  $\tilde{f}:\hat{\mathcal{K}} \rightarrow \mathbb{R}$. \\
\begin{algorithmic}[1]
\\ Sample $y_0$ from the uniform distribution on $B(x_0,r)\cap \hat{\mathcal{K}}$.

\For{$k=0$ to $k_{\mathrm{max}}$}
\\ Run Algorithm \ref{alg:SGLD} on $\hat{\mathcal{K}}$,  inverse temperature $\xi = \xi_k$, and step size $\eta_k$, $i_\mathrm{max}$, the oracle $g$, and the  initial point $X_0= y_k$.  Let $x_{k+1}$ be the output of Algorithm \ref{alg:SGLD}.\\
3. Sample $y_{k+1}$ from the uniform distribution on $B(x_{k+1}, r) \cap \hat{\mathcal{K}}$.
\EndFor
\end{algorithmic}
\textbf{output:} $x_{k_\mathrm{max}}$
\end{algorithm}

\subsection{Statement of Our Main Theorem} \label{sec:theorem}
We now formally state our main result, where we bound the error and running time when Algorithm \ref{alg:annealing} is used to solve Problem \ref{problem1}, assuming access to an oracle $\hat{F}$ that may be non-smooth or even non-continuous.

\begin{theorem} \label{thm:main} \textbf{(Main Theorem: Error bounds and running time for Algorithm \ref{alg:annealing})}
Let $F: \mathcal{K} \rightarrow \mathbb{R}$ be a convex function, and $\mathcal{K} \subseteq \mathbb{R}^d$ be a convex set that satisfy the assumptions stated in Section \ref{sec:notation}.
  Let $\hat{F}$ be a noisy oracle for $F$ with multiplicative noise of level $\alpha \leq O(1)$ and additive noise of level $\beta$, as in Equation \eqref{eq:assumption}.  
Let $\hat{\varepsilon} \geq 75\beta$ and  $\delta' >0$.
Then there exist parameters  $i_{\textrm{max}}$, $k_{\textrm{max}}$, $(\xi_k$, $\eta_k)_{k=0}^{k_{\mathrm{max}}}$, and $\sigma$, such that if we run Algorithm \ref{alg:annealing} with a smoothed version $\tilde{f}_\sigma$ of the oracle $\hat{F}$  (as defined in Section \ref{sec:smoothed_from_nonsmooth}), with probability at least  $1-\delta'$, the algorithm outputs a point $\hat{x}$ such that
\be
F(\hat{x}) -F(x^\star) &\leq \hat{\varepsilon},
\ee
with running time that is polynomial in $d$, $e^{\left(d\alpha + d\frac{\beta}{\hat{\varepsilon}} \right) \mathsf{c}}$, $R$, ${\lambda}$, $\frac{1}{r'}$, $\beta$, $\frac{1}{\hat{\varepsilon}}$, and $\log \frac{1}{\delta'}$, where $\mathsf{c}=O(\log(R+\lambda))$.
\end{theorem}

\noindent
We give a proof of Theorem \ref{thm:main} in Section \ref{sec:proof_of_main_result}.

\noindent
The precise values of the parameters in this theorem are quite involved and appear in the proofs at the following places:
$\xi_k$ appears in \eqref{eq:xi_k}, $\eta_k$ in \eqref{eq:eta_k}, $i_{\mathrm{max}}$ in \eqref{eq:i_max}, and the expression for $k_{\mathrm{max}}$ can be found in \eqref{eq:k_max}.
Below we present their approximate magnitudes. 
The inverse temperature parameter $\xi_k$, the smoothing parameter $\sigma$, and the number of epochs $k_{\mathrm{max}}$ satisfy:
\be
\tilde{\Omega}\left(\frac{d}{\lambda R} \right)\leq \xi_k \sim \tilde{O}\left(d \cdot \mathrm{max} \left\{ \frac{1}{\hat{\varepsilon}}, \hat{F}(X_0^{(0)}) \cdot 10^k\right\} \right) \leq \tilde{O}\left(\frac{d}{\hat{\varepsilon}}\right),
\ee
$$ \sigma  = \frac{1}{2} \min\left(\frac{\beta}{\lambda(1+\alpha)\sqrt{d}}, \frac{\mathsf{r}}{\sqrt{\log(\frac{1}{\alpha})+d}}\right),$$

\be
k_{\mathrm{max}} \sim \log \frac{R}{\hat{\varepsilon}}.
\ee

\noindent
To make the expressions for $\eta_k$ and $i_{\mathrm{max}}$ understandable,   assume that  $\lambda R > 1, \beta$, that  $\mathsf{r} > \frac{\beta}{\lambda}$, $d > \hat{\varepsilon}$, and that  $R> \lambda$. Then
\be
\eta_k \sim \frac{ \hat{\varepsilon}^4}{d^9 R^5 \lambda^8 \beta^4 }e^{-d\left(\alpha + \frac{\beta}{\hat{\varepsilon}}\right)\mathsf{c}'}  10^{-k}
\ee

\be
1 \leq i_{\mathrm{max}} \leq \left[\frac{d^{6.5}}{\hat{\varepsilon}^3} R^{\frac{11}{2}} \lambda^6 \beta e^{d\left( \alpha + \frac{\beta}{\hat{\varepsilon}}\right)\mathsf{c}'} \right ]^{1+\mathsf{c}''\alpha},
\ee

\noindent
where $\mathsf{c}' \sim \log\frac{R^2}{rd \delta \min\{ \hat{\varepsilon}/\lambda, r'\}}$
and $\mathsf{c}''$ is a constant factor.  In particular, the running time is given by $i_{\mathrm{max}} \times k_{\mathrm{max}}$.
\section{Proofs} \label{sec:proofs}

\subsection{Assumptions about the smooth oracle} \label{sec:smooth_assumptions}
We first show how to optimize $F$ if one has access to a smooth noisy objective function $\tilde{f}:\mathcal{K} \rightarrow \mathbb{R}$ (Sections \ref{sec:conductance}, \ref{sec:drift}, \ref{sec:result_smooth}).  Then, in Section \ref{sec:smoothing}, we show how one can obtain a smooth noisy objective function from a non-smooth and possibly non-continuos noisy objective function $\hat{F}$.
We will make the following assumptions  (we prove in Section \ref{sec:smoothing} that these assumptions hold for a smoothed version $\tilde{f}_\sigma$ of a non-smooth noisy objective function $\hat{F}$).
We assume the following noise model for $\tilde{f}$:
\be
\tilde{f}(x) = F(x)(1+\psi(x)) + \varphi(x),
\ee
for all $x \in \mathcal{K}$ where $|\psi(x)| \leq \alpha$ and $|\varphi(x)| \leq \beta$.  Note that, with a slight abuse of notation, in Section \ref{sec:smoothed_from_nonsmooth} we also used the letters $\alpha$ and $\beta$ to denote the noise levels of the non-smooth oracle $\hat{F}$, even though typically $\hat{F}$ will have lower noise levels than $\tilde{f}$.  In this section, as well as in Sections \ref{sec:conductance}-\ref{sec:result_smooth} where we assume direct access to a stochastic gradient for the smooth oracle $\tilde{f}$, we will instead refer to the noise levels of $\hat{F}$ by $``\hat{\alpha}"$ and $``\hat{\beta}"$. In Section \ref{sec:theorem}, on the other hand, $``\alpha"$ and $``\beta"$ will be used exclusively to denote the noise levels of $\hat{F}$.  We also assume that
\be \label{eq:assumption_noise}
\alpha \geq \hat{\alpha} \quad \textrm{and} \quad \beta \geq \hat{\beta}.
\ee

We make the following assumptions about $\tilde{f}$:
\begin{itemize}
\item $\psi(x)>-\alpha^{\dagger}$ for some $0\leq \alpha^{\dagger}<1$.  This assumption is needed because if not we might have $\psi(x) = -1$ for all $x\in \mathcal{K}$, in which case $\tilde{f}(x)$ would give no information about $F$.
\item $\|\nabla \tilde{f}(x) \| \leq \tilde{\lambda}$  for all $x \in  \mathcal{K}$.
\item We assume that we have access to a stochastic gradient $g$ such that $\nabla \tilde{f}(x) = \mathbb{E}[g(x)]$ for every $x \in \mathcal{K}$.  However, we do \emph{not} assume that we have oracle access to $\tilde{f}$ itself.
\end{itemize}

\begin{assumption}\label{assumption:A} ({\bf Based on assumption A in \cite{hitting_times}})
Let $\tilde{f}:\mathcal{K} \rightarrow \mathbb{R}^d$ be differentiable, and let $g \equiv g_w: \mathcal{K} \rightarrow \mathbb{R}^d$ be such that $\nabla \tilde{f}(x) = \mathbb{E}[g_W(x)]$ where $W$ is a random variable.
We will assume that
\begin{enumerate}
\item There exists $\zeta_\mathrm{max}>0$ such that for every compact convex $\hat{\mathcal{K}}\subseteq \mathbb{R}^d$, every $x\in \hat{\mathcal{K}}_{r'}$, and every $0\leq \zeta \leq \zeta_\mathrm{max}$, the random variable $Z \sim N(x,2\zeta I_d)$ satisfies $\mathbb{P}(Z \in \mathcal{K}) \geq \frac{1}{3}$. 
We prove this assumption in Lemma \ref{lemma:contraint_round}.
\item There exists $L>0$ such that $|\tilde{f}(y) -\tilde{f}(x) - \langle y-x, \nabla \tilde{f}(x)\rangle| \leq \frac{L}{2}\|y-x\|^2 \quad \quad \forall x,y \in \mathcal{K}$.
\item There exists $b_{\mathrm{max}}>0$ and $G>0$ such that for any $u \in \mathbb{R}^d$ with $\|u\| \leq b_{\mathrm{max}}$ the stochastic gradient $g(x)$ satisfies $\mathbb{E}[e^{\langle u, g(x)\rangle^2} |x] \leq e^{G^2\|u\|^2}$.
\end{enumerate}
\end{assumption}

\subsection{Conductance and bounding the Cheeger constant}
 \label{sec:conductance}

To help us bound the convergence rate, we define the Cheeger constant of a distribution, as well as the conductance of a Markov chain.
For any set $\hat{\mathcal{K}}$ and any function $f:\hat{\mathcal{K}} \rightarrow \mathbb{R}$, define
\be
\mu_f^{\hat{\mathcal{K}}}(x) := \frac{e^{-f(x)}}{\int_{\hat{\mathcal{K}}} e^{-f(x)}} \quad \quad  \forall x \in \hat{\mathcal{K}}.
\ee
 
\begin{definition}\label{def:Cheeger} (\textbf{Cheeger constant})
 For all $V \subseteq \hat{\mathcal{K}}$, define the Cheeger constant to be
\be
\mathcal{C}_{f}^{\hat{\mathcal{K}}}(V) := \liminf_{\varepsilon \downarrow 0} \inf_{A \subseteq V} \frac{\mu_f^{\hat{\mathcal{K}}}(A_\varepsilon) - \mu_f^{\hat{\mathcal{K}}}(A)}{\varepsilon \mu_f^{\hat{\mathcal{K}}}(A)}.
\ee
recalling that $A_\varepsilon = A+B(0,\varepsilon)$.
\end{definition}

\noindent For a Markov chain $Z_0,Z_1,\ldots$ on $\hat{\mathcal{K}}$ with stationary distribution $\mu_Z$ and transition kernel $Q_Z$, we define the conductance on a subset $V$ to be
\be
\Phi_Z^{\hat{^{\mathcal{K}}}}(V) := \inf_{A \subseteq V}\frac{\int_A  Q_Z(x, \hat{\mathcal{K}} \backslash A) \mu_Z(x) \mathrm{d}x}{\mu_Z(A)} \quad \quad \forall V \subseteq \hat{\mathcal{K}}
\ee
and the hitting time
\be
\tau_Z(A) := \inf\{i : Z_i \in A\} \quad \quad \forall A \subseteq \hat{\mathcal{K}}.
\ee
Finally,  we define the notion of two Markov chains being $\varepsilon'$-close:
\begin{definition}\label{def:epsilon_close}
If $W_0, W_1, \ldots$ and $Z_0, Z_1, \ldots$ are Markov chains on a set $\hat{\mathcal{K}}$ with transition kernels  $Q_W$ and $Q_Z$, respectively, we say that $W$ is $\varepsilon'$-close to $Z$ with respect to a set $U\subseteq \hat{\mathcal{K}}$ if 
\be
Q_Z(x,A) \leq Q_W(x,A) \leq (1+ \varepsilon') Q_Z(x,A)
\ee
for every $x \in \hat{\mathcal{K}} \backslash U$ and $A \subseteq \hat{\mathcal{K}} \backslash \{x\}$.
\end{definition}

\noindent
We now give a generalization of Proposition 2 in \cite{hitting_times}:
\begin{lemma} \label{lemma:cheeger2} \textbf{(Bounding the Cheeger constant)}
Assume that $\hat{\mathcal{K}} \subseteq \mathcal{K}'$ is convex, and that $F$ is convex and $\lambda$-Lipschitz on $\hat{\mathcal{K}}_{r'}$.  Then for every $\epsilon>0$ and all $\xi \geq \frac{4d\log(R/\min(\frac{\varepsilon}{2\lambda}, r'))}{\varepsilon}$ we have
\be
\mathcal{C}_{\xi {F}}^{\hat{\mathcal{K}}_{r'}}(\hat{\mathcal{K}}_{r'}\backslash U^\varepsilon) \geq \frac{1}{R}.
\ee
\end{lemma}

\begin{proof}
Let $\hat{x}^\star$ be a minimizer of ${F}$ on $\hat{\mathcal{K}}_{r'}$.  Let $\hat{r} = \min(\frac{\varepsilon}{2\lambda}, r')$.  
Then since $\hat{\mathcal{K}}_{r'} = \hat{\mathcal{K}}+B(0,r')$, for some $a \in \mathcal{K}'$ there is a closed ball $B(a,\hat{r}) \subseteq \hat{\mathcal{K}}_{r'}$, with $x^\star \in B(a,\hat{r})$.  
By the Lipschitz property, we have
\be
\sup \{{F}(x): x \in B(a,\hat{r})\} \leq {F}(x^\star) + 2\hat{r}\lambda \leq {F}(x^\star) + \frac{\varepsilon}{2}.
\ee

\noindent
Therefore, 
\be \label{eq:ratio4}
\inf \left \{\frac{e^{-\xi {F}(x)}}{e^{-\xi {F}(y)}} \, : \, x \in B(a,\hat{r}), y \in \mathcal{K}' \backslash U^\varepsilon) \right\} \geq e^{\xi\varepsilon/2}.
\ee

\noindent
   
\noindent
Then Equation \eqref{eq:ratio4} implies that
\be
\frac{\mu_{\xi {F}}^{\hat{\mathcal{K}}_{r'}}(\hat{\mathcal{K}}_{r'}\backslash U_\varepsilon)}{\mu_{\xi {F}}^{\hat{\mathcal{K}}_{r'}}(U_\varepsilon)} &\leq  e^{-\xi\varepsilon/2} \frac{\mathrm{Vol}(B(c,R))}{\mathrm{Vol}(B(a,\hat{r}))}\\
&= e^{-\xi\varepsilon/2} (\frac{R}{\hat{r}})^d\\
&= e^{-\xi\varepsilon/2 + d \log(R/\hat{r})}\\
&\leq \frac{1}{2},
\ee
which implies that
\be \label{eq:half}
\mu_{\xi {F}}^{\hat{\mathcal{K}}_{r'}}(\hat{\mathcal{K}}_{r'}\backslash U_\varepsilon)\leq \frac{1}{2}.
\ee

\noindent Then by Theorem 2.6 of \cite{lovasz1993random} for all $A \subseteq \hat{\mathcal{K}}_{r'}\backslash U^\varepsilon$ for any $0<\delta < 2R$ we have
\be
\mu_{\xi {F}}^{\hat{\mathcal{K}}_{r'}}(A_\delta \backslash A) &\geq \frac{2 \frac{\delta}{2R}}{1-\frac{\delta}{2R}} \min(\mu_{\xi {F}}^{\hat{\mathcal{K}}_{r'}}(A), \mu_{\xi {F}}^{\hat{\mathcal{K}}_{r'}}(\hat{\mathcal{K}}_{r'}\backslash A_\delta))\\
&= \frac{2 \frac{\delta}{2R}}{1-\frac{\delta}{2R}} \min(\mu_{\xi {F}}^{\hat{\mathcal{K}}_{r'}}(A), 1- \mu_{\xi {F}}^{\hat{\mathcal{K}}_{r'}}(A_\delta))\\
&= \frac{2 \frac{\delta}{2R}}{1-\frac{\delta}{2R}} \min(\mu_{\xi {F}}^{\hat{\mathcal{K}}_{r'}}(A), 1- \mu_{\xi {F}}^{\hat{\mathcal{K}}_{r'}}(A_\delta \backslash A)- \mu_{\xi {F}}^{\hat{\mathcal{K}}_{r'}}(A))\\
&\stackrel{{\scriptsize \textrm{Eq. }} \eqref{eq:half}}{\geq}  \frac{2 \frac{\delta}{2R}}{1-\frac{1}{2R}} \min(\mu_{\xi {F}}^{\hat{\mathcal{K}}_{r'}}(A), 1- \mu_{\xi {F}}^{\hat{\mathcal{K}}_{r'}}(A_\delta \backslash A)-\frac{1}{2})\\
&=  \frac{2 \frac{\delta}{2R}}{1-\frac{\delta}{2R}} \min(\mu_{\xi {F}}^{\hat{\mathcal{K}}_{r'}}(A), \frac{1}{2}- \mu_{\xi {F}}^{\hat{\mathcal{K}}_{r'}}(A_\delta \backslash A))\\
&=  \frac{2 \frac{\delta}{2R}}{1-\frac{\delta}{2R}} \mu_{\xi {F}}^{\hat{\mathcal{K}}_{r'}}(A),\\
\ee
provided that $0<\delta <\Delta_A$ for some small enough value $\Delta_A>0$ that depends on $A$.  Therefore for every $A \subseteq \hat{\mathcal{K}}_{r'}\backslash U_\varepsilon$ 
there exists $\Delta_A>0$ such that
\be
\frac{\mu_{\xi {F}}^{\hat{\mathcal{K}}_{r'}}(A_\delta \backslash A)}{\delta\mu_{\xi {F}}^{\hat{\mathcal{K}}_{r'}}(A)} \geq \frac{2\frac{1}{2R}}{1-\frac{\delta}{2R}} \quad \quad \forall 0<\delta <\Delta_A.
\ee 
Taking $\delta \rightarrow 0$, we get
\be
\mathcal{C}_{\xi {F}}^{\hat{\mathcal{K}}_{r'}}(\hat{\mathcal{K}}_{r'}\backslash U^\varepsilon) \geq  \frac{1}{R}.
\ee
\end{proof}

\subsection{Bounding the escape probability}
\label{sec:drift}
We will use the Lemma proved in this section (Lemma \ref{lemma:drift2}) to show that the SGLD chain $X$ defined in Algorithm \ref{alg:SGLD} does not drift too far from its initial objective function value with high probability.
This will allow us to bound the noise, since the noise is proportional to the objective function $F$.
The organization of this section is as follows: we first define a ``toy" algorithm and an associated Markov chain $Y$ that will allow us to prove Lemma \ref{lemma:drift2} (and which we will later use to prove Theorem \ref{alg:annealing}).  We then prove Propositions \ref{thm:drift} and \ref{lemma:drift}, and Lemma \ref{lemma:drift2}. Proposition \ref{thm:drift} is used to prove Proposition \ref{lemma:drift}, which in turn is used to prove Lemma \ref{lemma:drift2}.

We begin by recalling the Metropolis-adjusted version of Algorithm \ref{alg:SGLD} defined in \cite{hitting_times}, which defines a Markov chain $Y_0,Y_1,\ldots$ with stationary distribution $\mu_{\xi\tilde{f}}^{\mathcal{K}}$.  Note that this is a``toy" algorithm which is not meant to be implemented; rather we state this algorithm only to define the Markov chain $Y_0,Y_1,\ldots$, which we will use as a tool to prove Lemma \ref{lemma:drift2} and Theorem \ref{alg:annealing}.
\begin{algorithm}[H]
\caption{Lazy Metropolis-adjusted SGLD \label{alg:SGLD_metropolis}} 
\flushleft
\textbf{input:} Convex constraint set $\hat{\mathcal{K}} \subseteq \mathbb{R}^d$, inverse temperature $\xi>0$, step size $\eta > 0$, parameters $i_{\max} \in \mathbb{N}$ and $D>0$, stochastic gradient oracle $g$ for some $\tilde{f}:\mathcal{K} \rightarrow \mathbb{R}$, \\
 \textbf{input:}   Initial point $Y_0 \in \mathbb{R}^d$.\\
\begin{algorithmic}[1]
\For{$i=0$ to $i_{\mathrm{max}}$}
\\ Sample $P_i \sim N(0,I_d)$.
\\ Set $Y'_{i+1} = Y_i - \eta g(Y_i) + \sqrt{\frac{2\eta}{\xi}} P_i$.
\\ Set $Y''_{i+1} = Y'_{i+1}$ if $Y'_{i+1} \in \hat{\mathcal{K}} \cap B(Y_i, D)$.  Otherwise, set $Y''_{i+1} = Y_i$.
\\ Set $Y'''_{i+1} = Y''_{i+1}$ with probability $\min \left(1, \, \frac{\mathbb{E}[e^{-\frac{1}{4\eta}\|Y_i- Y''_{i+1} + \eta g(Y''_{i+1})\|}]}{\mathbb{E}[e^{-\frac{1}{4\eta}\|Y''_{i+1}- Y_{i} + \eta g(Y_{i})\|}]}
e^{\tilde{f}(Y_i)-\tilde{f}(Y''_{i+1})}\right)$.  Otherwise, set $Y'''_{i+1} = Y_i$.\\
\\ Set $V_i = 1$ with probability $\frac{1}{2}$ and set $V_i = 0$ otherwise. Let  $Y_{i+1} = Y'''_{i+1}$ if $V_i=1$; otherwise, let $Y_{i+1} = Y_i$.\\
\EndFor
\end{algorithmic}
\textbf{output:} $x_{i^\star}$, where $i^\star := \mathrm{argmin}_i \{\tilde{f}(x_i)\}$.
\end{algorithm}

\noindent We now define a coupling of three Markov chains.  We will use this coupling to prove Lemma \ref{lemma:drift2} and Theorem \ref{alg:annealing}.

\begin{definition} \label{def:coupling} \textbf{(Coupled Markov chains)}
Let $X$ and $\hat{X}$ be Markov chains generated by Algorithm \ref{alg:SGLD} with constraint set $\mathcal{K}$ and $\hat{\mathcal{K}}_{r'}$, respectively, where $\hat{\mathcal{K}}\subseteq \mathcal{K}$ and $\hat{\mathcal{K}}_{r'} = \hat{\mathcal{K}} + B(0,r')$.  Let $Y$ be the Markov chain generated by Algorithm \ref{alg:SGLD_metropolis}.
We define a coupling of the Markov chains $X$, $\hat{X}$ and $Y$ in the following way: 
 Define recursively, $t(0) = 0$, 
 $$t(i+1) = \min\{j\in \mathbb{N}: j>i, V_j = 1\}.$$
Let $Q_0,Q_1,\ldots \sim N(0,I_d)$ be i.i.d.  Let $X_0 = Y_0 = \hat{X}_0$. Let $Y_i$ be the chain in Algorithm \ref{alg:SGLD_metropolis} generated by setting $P_i = Q_i$ for all $i\geq 0$ with constraint set $\mathcal{K}$. 
Let $X$ be the chain in Algorithm \ref{alg:SGLD_metropolis} generated by setting $P_i = Q_{t(i)}$ for all $i\geq 0$ with constraint set $\mathcal{K}$.  Let $\hat{X}$ be the chain in Algorithm \ref{alg:SGLD_metropolis} generated by setting $P_i = Q_{t(i)}$ for all $i\geq 0$ with constraint set $\hat{\mathcal{K}}_{r'}$.
\end{definition}

\noindent We now bound the escape probability of the Markov chain $Y$ from a sub-level set of a given height, assuming that it is initialized from its stationary distribution conditioned on a small ball.

\begin{proposition} \label{thm:drift} \textbf{(Escape probability from stationary distribution on a small ball)}
Let $r>0$ be such that $r'\geq r>0$ and let $\xi>0$.  
 Let $Y_0, Y_1, \ldots$ be the Markov chain defined in Algorithm \ref{alg:SGLD_metropolis} with stationary distribution $\pi = \mu_{\xi \tilde{f}}^{\mathcal{K}}$, and let $Y_0$ be sampled from $\pi_0 := \mu_{\xi \tilde{f}}^{B(y,r) \cap \mathcal{K}}$, where $\pi_0$ is the distribution of $\pi$ conditioned on $B(y,r) \cap \mathcal{K}$ for some $y\in \mathcal{K}$.
Then for every $i \geq 0$ we have
\be
\mathbb{P}(\tilde{f}(Y_i) \geq h) \leq e^{\xi[\tilde{f}(y) + \tilde{\lambda}r] -\xi h + d\log(\frac{2R}{r})}  \quad \quad \forall h\geq 0.
\ee
\end{proposition}

\begin{proof}
Fix $h\geq 0$.  Define $S_1:= B(y,r) \cap \mathcal{K}$ and $S_2:= \{x \in \mathcal{K}: \tilde{f}(x) \geq h\}$.
Let $c_\pi = \int_{\mathcal{K}} e^{-\xi \tilde{f}(x)} \mathrm{d} x$ be the normalizing constant of $\pi$.  Since $\pi$ is the stationary distribution of $Y$,
\be \label{eq:drift2}
\mathbb{P}(Y_i \in S_2) \leq \frac{\pi(S_2)}{\pi(S_1)} \quad \quad \forall \, i \in \{0,\ldots,i_{\textrm{max}}\}
\ee
We can see why Inequality \eqref{eq:drift2} is true by the following argument:  Let $Z$ be a copy of the $Y$ chain started at stationarity.  Let $\mathcal{E}$ be the event that $Z_0 \in S_1$.  Then $Z_0|\mathcal{E}$ ($Z_0$ conditioned on the event $\mathcal{E}$) has the same distribution as $Y_0 \sim \pi_0$.  Therefore, $Z_i|\mathcal{E}$ has the same distribution as $Y_i$ (since the Z and Y chains have the same transition kernel).  Therefore, $\mathbb{P}(Z_i \in S_2 | \mathcal{E}) = \mathbb{P}(Y_i \in S_2)$.  Hence,
\be
\pi(S_2) &= \mathbb{P}(Z_i \in S_2) \geq \mathbb{P}(\{Z_i \in S_2\} \cap \mathcal{E}) = \mathbb{P}(Z_i \in S_2 | \mathcal{E}) \mathbb{P}(\mathcal{E}) =  \mathbb{P}(Y_i \in S_2 ) \mathbb{P}(Z_0 \in S_1) =\mathbb{P}(Y_i \in S_2 ) \pi(S_1),
\ee
 which implies Inequality \eqref{eq:drift2}.

But $\|\nabla \xi \tilde{f} \|=\| \xi \nabla \tilde{f}\| \leq \xi \tilde{\lambda}$, implying that
\be  \label{eq:drift3}
\pi(S_1) = \pi(B(y,r)) &\geq \; c_\pi e^{-[\xi \tilde{f}(y) +\xi \tilde{\lambda}r]} \times \mathrm{Vol}(B(y,r) \cap \mathcal{K})\\
& \geq c_\pi e^{-\xi[\tilde{f}(y) +\tilde{\lambda}r]} \times \mathrm{Vol}(B(0,\frac{1}{2}r)),
\ee
since $B(y,r) \cap \mathcal{K}$ contains a ball of radius $\frac{1}{2}r$ because $r\leq r'$.
Also,
\be  \label{eq:drift4}
\pi(S_2) = \pi(\{x: \tilde{f}(x) \geq h\}) &\leq c_\pi e^{-\xi h} \mathrm{Vol}(\mathcal{K})
\leq c_\pi e^{-\xi h} \mathrm{Vol}(B(0,R)).
\ee
Therefore,
\be
\mathbb{P}(Y_i \in S_2) &\stackrel{{\scriptsize \textrm{Eq. }} \eqref{eq:drift2}}{\leq} \frac{\pi(S_2)}{\pi(S_1)}\\
&\stackrel{{\scriptsize \textrm{Eq. }} \eqref{eq:drift3}, \, \eqref{eq:drift4}}{\leq}  e^{\xi [\tilde{f}(y) +\tilde{\lambda}r]-\xi h} \times \left(\frac{R}{\frac{1}{2}r}\right)^d\\
&= e^{\xi[\tilde{f}(y) +\tilde{\lambda}r]-\xi h +d \log(\frac{2R}{r})}.
\ee
\end{proof}

\noindent We now extend our bound for the escape probability of the Markov chain $Y$ (Proposition \ref{thm:drift}) to the case where $Y$ is instead initialized from the \emph{uniform} distribution on a small ball:

\begin{proposition} \label{lemma:drift} \textbf{(Escape probability from uniform distribution on a small ball)}
Let $r>0$ be such that $r'\geq r>0$ and let $\xi>0$. Let $\nu_0$ be the uniform distribution on $B(y,r) \cap \mathcal{K}$ for some $y\in \mathcal{K}$. 
Let $Y_0, Y_1, \ldots$ be the Markov chain defined in Algorithm \ref{alg:SGLD_metropolis} with stationary distribution $\pi = \mu_{\xi \tilde{f}}^{\mathcal{K}}$, and let $Y_0$ be sampled from $\nu_0$.
Then for every $i \geq 0$ we have
\be \label{eq:drift1}
\mathbb{P}(\tilde{f}(Y_i) \geq h) &\leq e^{\xi[\tilde{f}(y) + \tilde{\lambda}r] -\xi h + d\log(\frac{2R}{r})}  + 2r\tilde{\lambda}\xi \quad \quad \forall h\geq 0.
\ee
Moreover, for every $A \subseteq \mathcal{K}$, we have
\be  \label{eq:warmstart}
 \nu_0(A) \leq e^{2R\tilde{\lambda}\xi + d\log(\frac{2R}{r})}  \pi(A).
\ee

\end{proposition}

\begin{proof}
Since $\|\nabla \xi \tilde{f}(x) \| \leq \xi \tilde{\lambda}$,
\be
\sup_{x\in B(y,r)\cap \mathcal{K}} \xi \tilde{f}(x) - \inf_{x\in B(y,r)\cap \mathcal{K}} \xi \tilde{f}(x) \leq 2r\tilde{\lambda} \xi,
\ee
and hence
\be \label{eq:ratio}
\frac{\inf_{x\in B(y,r)\cap \mathcal{K}} \pi(x)}{\sup_{x\in B(y,r)\cap \mathcal{K}} \pi(x)} \geq e^{-2r\tilde{\lambda}\xi}.
\ee

\noindent
Define  $\pi_0 := \mu_{\xi \tilde{f}}^{B(y,r) \cap \mathcal{K}}$ to be the distribution of $\pi$ conditioned on $B(y,r) \cap \hat{\mathcal{K}}'$.  Let $Z$ be sampled from the distribution $\pi_0$. 
 Let $Z'= Y_0$ with probability $\min(\frac{\pi_0(Y_0)}{\nu_0(Y_0)},1)$; otherwise let $Z'=Z$.  Then $Z'$ has distribution $\pi_0$.  
 Moreover, by Equation \eqref{eq:ratio}, $Z'= Y_0$ with probability at least $e^{-2r\tilde{\lambda}\xi}$.  Therefore, by Proposition \ref{thm:drift}
\be
\mathbb{P}(\tilde{f}(Y_i) \geq h) &\leq e^{\xi[\tilde{f}(y) + \tilde{\lambda}r] -\xi h + d\log(\frac{2R}{r})}  + 1-e^{-2r\tilde{\lambda}\xi}\\
&\leq e^{\xi[\tilde{f}(y) + \tilde{\lambda}r] -\xi h + d\log(\frac{2R}{r})}  + 2r\tilde{\lambda}\xi \quad \quad \forall h\geq 0.
\ee
This proves Equation \eqref{eq:drift1}.
Now, since $\|\xi \nabla \tilde{f}(x) \| \leq \xi \tilde{\lambda}$ and $\mathcal{K} \subseteq B(c,R)$,
\be
\sup_{x\in \mathcal{K}} \xi \tilde{f}(x) - \inf_{x\in \mathcal{K}} \xi \tilde{f}(x) \leq 2R\tilde{\lambda} \xi,
\ee
implying that
\be \label{eq:ratio2}
\frac{\inf_{x\in \mathcal{K}} \pi(x)}{\sup_{x\in \mathcal{K}} \pi(x)} \geq e^{-2R\tilde{\lambda}\xi}.
\ee
Therefore, for every $z \in \mathcal{K}$ we have
\be \label{eq:ratio3}
\frac{\pi(z)}{ \nu_0(z)} &= \mathrm{Vol}(B(y,r)\cap \mathcal{K})\times \pi(z)\\
&\geq \mathrm{Vol}(B(0,\frac{1}{2}r))\times \pi(z)\\
&\geq \mathrm{Vol}(B(0,\frac{1}{2}r)) \times \frac{1}{\mathrm{Vol}(B(0,2R))}\frac{\inf_{x\in \mathcal{K}} \pi(x)}{\sup_{x\in \mathcal{K}} \pi(x)}\\
&\stackrel{{\scriptsize \textrm{Eq. }} \eqref{eq:ratio2}}{\geq} \left(\frac{2R}{r}\right)^{-d} e^{-2R\tilde{\lambda}\xi}\\
&= e^{-2R\tilde{\lambda}\xi - d\mathrm{log}(\frac{2R}{r})}. 
\ee
Where the second inequality holds since $r\leq r'$.  This proves Equation \eqref{eq:warmstart}.
\end{proof}

\noindent We are now ready to bound the escape probability of the SGLD Markov chain $X$ defined in Algorithm \ref{alg:SGLD} when it is initialized from the uniform distribution on a small ball:

\begin{lemma} \label{lemma:drift2} \textbf{(Escape probability for unadjusted SGLD chain)}  Let $r>0$ be such that $r'\geq r>0$ and let $\xi>0$. 
Let $\nu_0$ be the uniform distribution on $B(y,r) \cap \mathcal{K}$ for some $y\in \mathcal{K}$, and let $X_0$ be sampled from $\nu_0$. 
 Let $X_0,X_1, \ldots$ be the Markov chain generated by Algorithm \ref{alg:SGLD} with constraint set $\mathcal{K}$.
  Let $\delta \leq \frac{1}{4}$ and let $0<\eta \leq \frac{\delta}{{i_\mathrm{max}}\times 16 d(G^2+L)}$  then
\be
\mathbb{P}(\tilde{f}(X_i) \geq h) &\leq e^{\xi[\tilde{f}(y) + \tilde{\lambda}r] -\xi h + d\log(\frac{2R}{r})}  + 2r\tilde{\lambda}\xi + \delta \quad \quad \forall h\geq 0.
\ee
\end{lemma}
\begin{proof}

Let $Y_0, Y_1, \ldots$ be the Markov chain generated by Algorithm \ref{alg:SGLD_metropolis}, and let $X_0, X_1, \ldots$  be the Markov chain defined in Algorithm \ref{alg:SGLD}, where both chains have constraint set $\mathcal{K}$. Couple the Markov chains $X$ and $Y$ as in Definition \ref{def:coupling}.
By Claim 2 in the proof of Lemma 13 of \cite{hitting_times}, for each $i \geq 0$ the rejection probability $\mathbb{P}(Y_{i+1} = Y_i)$ is bounded above by $1-e^{-16 \eta d (G^2+L)} \leq 1- e^{-\frac{\delta}{{i_\mathrm{max}}}} \leq \frac{\delta}{{i_\mathrm{max}}}$. 
 Hence, for all $0\leq i \leq i_\mathrm{max}$ we have
\be \label{eq:coupling}
\mathbb{P}(X_j =Y_j \, \, \forall 0\leq j \leq i) \geq (1- \frac{\delta}{{i_\mathrm{max}}})^{i} \geq 1-\delta.
\ee
Thus,
\be
\mathbb{P}(\tilde{f}(X_i) \geq h) &\leq \mathbb{P}(\tilde{f}(Y_i) \geq h) + \mathbb{P}(X_j \neq Y_j \textrm{ for some } 0\leq j\leq i)\\
&\stackrel{{\scriptsize \textrm{ Eq. }}\eqref{eq:coupling}}{\leq} \mathbb{P}(\tilde{f}(Y_i) \geq h) + \delta \quad \quad \forall h\geq 0\\
&\stackrel{{\scriptsize \textrm{Proposition }} \ref{lemma:drift}}{\leq} e^{\xi[\tilde{f}(y) + \tilde{\lambda}r] -\xi h + d\log(\frac{2R}{r})}  + 2r\tilde{\lambda}\xi + \delta \quad \quad \forall h\geq 0.
\ee
\end{proof}

\subsection{Comparing noisy functions} \label{sec:compare}
In this section we bound the ratio of $\hat{F}$ to $\tilde{f}$. We use this bound to prove Theorem \ref{alg:annealing} in Section \ref{sec:result_smooth}.
\begin{lemma} \label{lemma:compare} (\textbf{Bounding the ratio of two noisy objective functions}) 
Fix $x \in \mathcal{K}$ and let $t \geq 5\beta$. Define $\hat{H} = \max\{ \tilde{f}(x), t\}$ and let $\hat{J} = \max\{ \hat{F}(x), t\}$. Then,
\be
\frac{1}{5} \hat{H} \leq \hat{J}\leq 5 \hat{H}.
\ee

\end{lemma}
\begin{proof}
By our assumption in Equation \eqref{eq:assumption_noise}, we have that
 \be
 |F(x)-\hat{F}(x)| \leq \hat{\alpha} F(x) + \hat{\beta} \leq \alpha F(x) + \beta.
 \ee
Since $\alpha <\frac{1}{2}$, we have,
 \be \label{eq:b1}
 F(x) \leq 2 \hat{F}(x) + 2 \beta.
 \ee

\noindent We also have that,
 \be
 |\tilde{f}(x)-F(x)| \leq \alpha F(x) + \beta
 \ee
implying that
 \be \label{eq:b2}
 \tilde{f}(x) \leq 4 F(x) + \beta.
 \ee
Therefore, combining Equations \eqref{eq:b1} and \eqref{eq:b2}, we have
\be \label{eq:b3}
 \tilde{f}(x) \leq 4 \hat{F}(x) + 5\beta,
 \ee
implying that
\be
 \max(\tilde{f}(x), 5 \beta) \leq \max(4 \hat{F}(x) + 5\beta, 20 \beta).\\
 \ee
 Thus,
 \be
 \max(\tilde{f}(x), 5 \beta) \leq 5\max(\hat{F}(x),5 \beta).
\ee
Thus, we have $\hat{H}\leq 5\hat{J}$.  By a similar argument as above, we can also show that $\hat{J} \leq 5\hat{H}$.
\end{proof}

\subsection{Bounding the error and running time: The smooth case}
 \label{sec:result_smooth}

In this section we will show how to bound the error and running time of Algorithm \ref{alg:SGLD}, if we assume that we have access to a stochastic gradient oracle  $g$ for a smooth noisy function $\tilde{f}$, which approximates the convex function $F$.  In particular, we do \emph{not} assume access to the smooth function $\tilde{f}$ itself, only to $g$.  We also assume access to a non-smooth oracle $\hat{F}$, which we use to determine the temperature parameter for our Markov chain based on the value of $\hat{F}(X_k^{0})$ at the beginning of each epoch.  To prove the running time and error bounds, we will use the results of Sections \ref{sec:conductance} and \ref{sec:drift}.

Recall that in this section $\alpha$ and $\beta$ refer exclusively to the multiplicative and additive noise levels of $\tilde{f}$.
  We must first define parameters that will be needed to formally state and prove our error and running time bounds:

\begin{itemize}
\item Fix $0\leq \varepsilon < \frac{1}{25}$ and $\delta>0$.  

\item Set parameters of Algorithms \ref{alg:SGLD} and \ref{alg:annealing} as follows: 

\begin{itemize}
\item Let $y_0\in \mathcal{K}$ and let $H_0 := \tilde{f}(y_0)$. 

\item Fix $\mathfrak{D} \geq \frac{1}{\varepsilon}\beta$.  For every $0\leq k \leq k_{\mathrm{max}}$, let $H_k := \tilde{f}(x_k)$ and define $\hat{H}_k := \max(H_k, \mathfrak{D})$. 

\item Assume, without loss of generality, that $r' \leq \frac{\mathfrak{D}}{\lambda}$.  \footnote{This is without loss of generality since if there exists a convex body $\mathcal{K}'$ such that $\mathcal{K}' +B(0,r') = \mathcal{K}$, then for every $0 < \rho \leq r'$ there must also exist a convex body $\mathcal{K}''$ such that $\mathcal{K}'' +B(0,\rho) = \mathcal{K}$, namely $\mathcal{K}'' = \mathcal{K}' + B(0,r'-\rho)$.}

\item For every $0\leq k \leq k_{\mathrm{max}}$, let $J_k := \hat{F}(x_k)$.  Define $\hat{J}_k := \max(J_k, \mathfrak{D})$.

\item Set the number of epochs to be $k_{\mathrm{max}}= \lceil \frac{\log(5 J_0/\mathfrak{D})}{\log(\frac{1}{25 \varepsilon})}\rceil+1$.

\item At every $k \geq 0$, set the temperature to be $\xi_k =  \frac{4 d \log(R/\min(\frac{\varepsilon}{2\lambda} \mathfrak{D}, r'))}{\frac{1}{5} \varepsilon \hat{J}_k}$.  Define $\bar{\xi} := \frac{4 d \log(R/\min(\frac{\varepsilon}{2\lambda}  \mathfrak{D}, r'))}{\frac{1}{25}\varepsilon  \mathfrak{D}}$.  

\item Set $r = \frac{\delta}{\bar{\xi}\tilde{\lambda}}$.

\item
Define
\be
\bar{\eta}^{\dagger} := c\min \left\{\zeta_{\mathrm{max}}, d\frac{\omega^2}{\lambda^2}, \frac{b_{\mathrm{max}}^2}{d}, \frac{1}{Rd^3((\bar{\xi} G)^2 +\bar{\xi} L)^2}\right\}
\ee
and
\be
\mathfrak{B}':= \frac{\left(d\log(2\frac{R}{r}) + \delta + 1+ \log(\frac{1}{\delta})\right))}{2 d \log(R/\min(\frac{\varepsilon}{2\lambda} \mathfrak{D}, r'))}.
\ee

 \item Set the number of steps $i_{\mathrm{max}}$ for which we run the the Markov chain $X$ in each epoch to be
 \be
 i_\mathrm{max}=\left \lceil \left( \frac{8R\tilde{\lambda}\xi_k + 4d(1+ \log(1+ \bar{\xi}) + \log(\frac{2R\tilde{\lambda}}{\delta})) + 4\log(\frac{1}{\delta})}{\left(\frac{1}{1536R}\sqrt{\bar{\eta}^{\dagger}/d} e^{-\frac{150d}{\varepsilon} \left[\frac{\alpha}{1-\alpha^{\dagger}}\left(3+\varepsilon \mathfrak{B}' + \frac{\beta }{\mathfrak{D}}\right) + \frac{\beta}{\mathfrak{D}} \right]\log(R/\min(\frac{\varepsilon}{2\lambda} \mathfrak{D}, r'))}\right)^2}\right)^{\frac{1}{1-\frac{150}{\varepsilon}\alpha}}\right\rceil  +1 .
 \ee
\item Define $\mathfrak{B}:= \frac{\left(d\log(2\frac{R}{r}) + \delta +  \log(i_{\max}+1) + \log(\frac{1}{\delta})\right))}{2 d \log(R/\min(\frac{\varepsilon}{2\lambda} \mathfrak{D}, r'))}$.

\item For every $\xi >0$ define
\be
\eta(\xi) := c\min \left\{\zeta_{\mathrm{max}}, d\frac{\omega^2}{\lambda^2}, \frac{b_{\mathrm{max}}^2}{d}, \frac{(e^{-\frac{100d}{\varepsilon} \left[\frac{\alpha}{1-\alpha^{\dagger}}\left(3+\varepsilon \mathfrak{B} + \frac{\beta}{\mathfrak{D}}\right) + \frac{\beta}{\mathfrak{D}} \right]\log(R/\min(\frac{\varepsilon}{2\lambda} \mathfrak{D}, r'))})^2}{Rd^3((\xi G)^2 +\xi L)^2}\right\},
\ee
\noindent where $\omega = \varepsilon \mathfrak{D}$, and $c$ is the universal constant in Lemma 15 of \cite{hitting_times}.

\noindent Set the step size at each epoch to be $\eta_k = \eta(\xi_k)$.  Also define $\bar{\eta} = \eta(\bar{\xi})$.
 
 \item Set $D = \sqrt{2\bar{\eta} d}$.
\end{itemize}
\end{itemize}

\noindent
We now state the error and running time bounds:

\begin{theorem} \label{thm:error} \textbf{(Error and running time bounds when using a smooth noisy objective function)}
Assume that $\alpha \leq \frac{\varepsilon}{32}$.
Then with probability at least $1-6\delta (k_\mathrm{max}+1)$ Algorithm \ref{alg:annealing} returns a point $\hat{x}=x_{k_\mathrm{max}}$ such that
\be
F(\hat{x}) -F(x^\star) &\leq \frac{1}{1-\alpha}(\mathfrak{D} + \beta),
\ee
with running time that is polynomial in $d$, $e^{\frac{d}{\varepsilon/150} \left[\frac{\alpha}{1-\alpha^{\dagger}}\left(3+\varepsilon \mathfrak{B}' + \frac{\beta}{\mathfrak{D}}\right) + \frac{\beta}{\mathfrak{D}} \right] \log(R/\min(\frac{\varepsilon}{2\lambda} \mathfrak{D}, r'))}$, $R$, ${\lambda}$, ${\tilde{\lambda}}$, $L$, G, $\zeta_{\mathrm{max}}$, $b_{\mathrm{max}}$, and $\log(\frac{1}{\delta})$.
\end{theorem}
\begin{proof}
Set notation as in Algorithms \ref{alg:SGLD} and \ref{alg:annealing}.  Denote by $X^{(k)}$ the Markov chain generated by Algorithm \ref{alg:SGLD} as a subroutine in the $k$'th epoch of Algorithm \ref{alg:annealing} with constraint set $\mathcal{K}$.

Set $h_k=\hat{H}_k + \xi_k^{-1}\left(d\log(\frac{2R}{r})  + \delta +  \log(i_{\max}+1) + \log(\frac{1}{\delta})\right).$
 Then by Lemma \ref{lemma:drift2}
\be \label{eq:delta2}
\mathbb{P}( \sup_{0\leq i \leq i_{\max}}\tilde{f}(X_i^{(k)}) \geq h_k) &\leq (i_{\max}+1)\times[ e^{\xi_k[\hat{H}_k + \tilde{\lambda}r] -\xi_kh_k + d\log(\frac{2R}{r})}  + 2r\tilde{\lambda}\xi_k + \delta]\\
&\leq e^{\xi_k \hat{H}_k  + \delta -\xi_k h_k + d\log(\frac{2R}{r}) +  \log(i_{\max}+1)}  + 4\delta]\\
&=   5\delta,
\ee
where the second inequality holds since $r = \frac{\delta}{\bar{\xi} \tilde{\lambda}}$ and $\xi_k \leq \bar{\xi}$ for all $k$.

\noindent
But $\tilde{f}(X_i^{(k)}) \geq h_k$ if and only if
\be
F(X_i^{(k)})(1+\psi(X_i^{(k)})) + \varphi(X_i^{(k)}) \geq h_k
\ee
if and only if
\be
F(X_i^{(k)}) \geq \frac{1}{1+\psi(X_i^{(k)})}(h_k - \varphi(X_i^{(k)})),
\ee
since $1+\psi(X_i^{(k)}) \geq 0$.  Also,
\be
\frac{1}{1+\psi(X_i^{(k)})}(h_k - \varphi(X_i^{(k)})) \leq \frac{1}{1-\alpha^{\dagger}}(h_k + \beta),
\ee
since $\psi(X_i^{(k)}) \geq -\alpha^{\dagger}>-1$ and $|\varphi(X_i^{(k)})|< \beta$.

\noindent Hence, 
\be \label{eq:delta}
5 \delta &\stackrel{{\scriptsize \textrm{Eq. }} \eqref{eq:delta2}}{\geq} \mathbb{P}\left(\sup_{0\leq i \leq i_{\max}}\tilde{f}(X_i^{(k)}) \geq h_k\right)\\
&\geq \mathbb{P}\left(\sup_{0\leq i \leq i_{\max}} F(X_i^{(k)}) \geq \frac{1}{1-\alpha^{\dagger}}(h_k + \beta)\right).\\
\ee

\noindent
Define $\hat{\mathcal{K}}^{(k)} := (\mathcal{K}' \cap \{x \in \mathbb{R}^d: F(x) \leq \frac{1}{1-\alpha^{\dagger}}(h_k + \beta) + \lambda r'\}) + B(0,r')$.  Then
\be \label{eq:subset}
\left \{x\in \mathcal{K}: F(x) \leq \frac{1}{1-\alpha^{\dagger}}(h_k + \beta) \right\} \subseteq \hat{\mathcal{K}}^{(k)},
\ee
since $\|\nabla F\| \leq \lambda$.
Thus, by Equations \eqref{eq:delta} and \eqref{eq:subset},
\be \label{eq:delta3}
\mathbb{P}\left(X_i^{(k)}\in \hat{\mathcal{K}}^{(k)} \, \, \forall 0\leq i \leq i_{\max}\right) \geq 1-5 \delta.
\ee

\noindent
Also, for every $x \in \hat{\mathcal{K}}^{(k)}$, since $r' \leq \frac{\mathfrak{D}}{\lambda}$, we have
\be \label{eq:F}
F(x) &\leq \frac{1}{1-\alpha^{\dagger}}(h_k + \beta) + 2 \lambda r'\\
&\leq \frac{1}{1-\alpha^{\dagger}}(h_k + \beta) + 2\mathfrak{D}\\
&= \frac{1}{1-\alpha^{\dagger}}\left(\hat{H}_k + \xi_k^{-1}\left(d\log(4\frac{R}{r}) + \delta +  \log(i_{\max}+1) + \log(\frac{1}{\delta})\right) + \beta \right) + 2\mathfrak{D}\\
&= \frac{1}{1-\alpha^{\dagger}}\left(\hat{H}_k + \frac{1}{5}\varepsilon \hat{J}_k \frac{\left(d\log(4\frac{R}{r}) + \delta +  \log(i_{\max}+1) + \log(\frac{1}{\delta})\right))}{2 d \log(R/\min(\frac{\varepsilon}{2\lambda} \mathfrak{D}, r'))}  + \beta \right) + 2\mathfrak{D}\\
&\stackrel{{\scriptsize \textrm{Lemma }} \ref{lemma:compare}}{\leq} \frac{1}{1-\alpha^{\dagger}}\left(\hat{H}_k + \varepsilon \hat{H}_k \mathfrak{B}  + \beta \right) + 2\mathfrak{D}.
\ee
Thus, for every $x \in \hat{\mathcal{K}}^{(k)}$,
\be \label{eq:noise}
|N(x)| &\leq \alpha F(x) + \beta\\
&\stackrel{{\scriptsize \textrm{Eq. }} \eqref{eq:F}}{\leq} \frac{\alpha}{1-\alpha^{\dagger}}\left(\hat{H}_k(1+\varepsilon \mathfrak{B}) + \beta + 2\mathfrak{D}\right) +\beta :=N_k.
\ee

\noindent
Define $U_k^{\varepsilon''} := \{x \in \hat{\mathcal{K}}^{(k)} : F(x) \leq \varepsilon''\}$ for every $\varepsilon'' >0$. Then by Lemma \ref{lemma:cheeger2}, 
$\mathcal{C}_{(\xi F)}^{\hat{\mathcal{K}}^{(k)}}(\hat{\mathcal{K}}^{(k)}\backslash U_k^{\varepsilon \hat{H}_k}) \geq \frac{1}{R}$ for any $\xi \geq \frac{4 d \log(R/\min(\frac{\varepsilon}{2\lambda} \hat{H}_k, r')}{\varepsilon \hat{H}_k}$.

\noindent
 But by Lemma \ref{lemma:compare}, $\xi_k =  \frac{4 d \log(R/\min(\frac{\varepsilon}{2\lambda} \mathfrak{D}, r'))}{\frac{1}{5} \varepsilon \hat{J}_k} \geq  \frac{4 d \log(R/\min(\frac{\varepsilon}{2\lambda} \hat{H}_k, r')}{\varepsilon \hat{H}_k}$, implying that
 \be \label{eq:conductance}
\mathcal{C}_{(\xi_k \tilde{f})}^{\hat{\mathcal{K}}^{(k)}}(\hat{\mathcal{K}}^{(k)}\backslash U_k^{\varepsilon \hat{H}_k}) &\stackrel{{\scriptsize \textrm{Eq. }} \eqref{eq:noise}}{\geq} e^{-2\xi_k N_k} \mathcal{C}_{(\xi_k F)}^{\hat{\mathcal{K}}^{(k)}}(\hat{\mathcal{K}}^{(k)}\backslash U_k^{\varepsilon \hat{H}_k})\\
&\stackrel{{\scriptsize \textrm{Lemma }} \ref{lemma:cheeger2}}{\geq} \frac{1}{R}e^{- 2\xi_k N_k}\\
&= \frac{1}{R}e^{-\frac{4d}{\varepsilon} \frac{N_k}{\frac{1}{5}\hat{J}_k}\log(R/\min(\frac{\varepsilon}{2\lambda}\frac{1}{5} \mathfrak{D}, r'))}\\
&\stackrel{{\scriptsize \textrm{Lemma }} \ref{lemma:compare}}{\geq} \frac{1}{R}e^{-\frac{4d}{\varepsilon} \frac{N_k}{\frac{1}{25}\hat{H}_k}\log(R/\min(\frac{\varepsilon}{2\lambda} \mathfrak{D}, r'))}\\
&= \frac{1}{R}e^{-\frac{4d}{\varepsilon} \frac{\frac{\alpha}{1-\alpha^{\dagger}}\left(\hat{H}_k(1+\varepsilon \mathfrak{B}) + 2\mathfrak{D}+ \beta\right)+ \beta}{\frac{1}{25}\hat{H}_k}\log(R/\min(\frac{\varepsilon}{2\lambda} \mathfrak{D}, r'))}\\
&\geq \frac{1}{R}e^{-\frac{100d}{\varepsilon} \left[\frac{\alpha}{1-\alpha^{\dagger}}\left(3+\varepsilon \mathfrak{B} + \frac{\beta}{\mathfrak{D}}\right) + \frac{\beta}{\mathfrak{D}} \right]\log(R/\min(\frac{\varepsilon}{2\lambda} \mathfrak{D}, r'))},\\
\ee
where the first inequality holds by the stability property of the Cheeger constant, and the last inequality is true since $\hat{H}_k \geq \mathfrak{D}$ by definition.

\noindent
Recall that
\be \label{eq:eta}
\eta_k &= c\min \left\{\zeta_{\mathrm{max}}, d\frac{\omega^2}{\lambda^2}, \frac{b_{\mathrm{max}}^2}{d}, \frac{(e^{-\frac{100d}{\varepsilon}  \left[\frac{\alpha}{1-\alpha^{\dagger}}\left(3+\varepsilon \mathfrak{B} + \frac{\beta }{\mathfrak{D}}\right) + \frac{\beta}{\mathfrak{D}} \right]\log(R/\min(\frac{\varepsilon}{2\lambda} \mathfrak{D}, r'))})^2}{Rd^3((\xi_k G)^2 +\xi_kL)^2}\right\}\\
&\stackrel{{\scriptsize \textrm{Eq. }} \eqref{eq:conductance}}{\leq} c\min \left\{ \zeta_{\mathrm{max}}, d\frac{\omega^2}{\lambda^2}, \frac{b_{\mathrm{max}}^2}{d}, \frac{(\mathcal{C}_{(\xi_k \tilde{f})}^{\hat{\mathcal{K}}^{(k)}}(\hat{\mathcal{K}}^{(k)}\backslash U_k^{\varepsilon \hat{H}_k}) )^2}{d^3((\xi_k G)^2 +\xi_kL)^2} \right\},
\ee
where $\omega = \varepsilon \mathfrak{D}$. 

Recall that $X^{(k)}$ is the subroutine Markov chain described in Algorithm \ref{alg:SGLD} with inputs specified by Algorithm \ref{alg:annealing} and constraint set $\mathcal{K}$.
Let $\hat{X}^{(k)}$ be the Markov chain generated by Algorithm \ref{alg:SGLD} with constraint set $\hat{\mathcal{K}}_{r'}^{(k)}$ and initial point $X_0^{(k)} = \hat{X}_0^{(k)}$.
Let $Y^{(k)}$ be the Markov chain generated by Algorithm \ref{alg:SGLD_metropolis} with constraint set $\hat{\mathcal{K}}_{r'}^{(k)}$.  Couple the Markov chains as in definition \ref{def:coupling}.

\noindent
Write
\be
(U_k^{\varepsilon \hat{H}_k})_{\omega/\lambda} := (U_k^{\varepsilon \hat{H}_k} + B(0,\omega/\lambda)) \cap \hat{\mathcal{K}}^{(k)}
\ee
as shorthand.
Then by Lemma 15 of \cite{hitting_times} and by Equation \eqref{eq:eta}, the Markov chain $\hat{X}^{(k)}$ is $\varepsilon'$-close to  $Y^{(k)}$ with $\varepsilon' \leq \frac{1}{4} \Phi_Y(\hat{\mathcal{K}}^{(k)}\backslash (U_k^{\varepsilon \hat{H}_k})_{\omega/\lambda})$ and
\be \label{eq:conductance2}
\Phi_Y(\hat{\mathcal{K}}^{(k)}\backslash (U_k^{\varepsilon \hat{H}_k})_{\omega/\lambda})) &\geq \frac{1}{1536}\sqrt{\eta_k/d}\mathcal{C}_{(\xi_k \tilde{f})}^{\hat{\mathcal{K}}^{(k)}}(\hat{\mathcal{K}}^{(k)}\backslash (U_k^{\varepsilon \hat{H}_k})_{\omega/\lambda}) \\
&\stackrel{{\scriptsize \textrm{Eq. }} \eqref{eq:conductance}}{\geq} \frac{1}{1536R}\sqrt{\eta_k/d} e^{-\frac{100d}{\varepsilon} \left[\frac{\alpha}{1-\alpha^{\dagger}}\left(3+\varepsilon \mathfrak{B} + \frac{\beta }{\mathfrak{D}}\right) + \frac{\beta}{\mathfrak{D}} \right]\log(R/\min(\frac{\varepsilon}{2\lambda} \mathfrak{D}, r'))}.
\ee

\noindent 
Recall that by Equation \eqref{eq:warmstart} of Proposition \ref{lemma:drift}, for every $A \subseteq \mathcal{K}'$, we have
\be
\nu_0(A) \leq e^{4R\tilde{\lambda} \xi_k + d\mathrm{log}(\frac{2R}{r})}  \mu^{\hat{\mathcal{K}}^{(k)}}_{\xi_k\tilde{f}}(A).
\ee
Therefore, since $\hat{X}^{(k)}$ is $\varepsilon'$-close to $Y^{(k)}$, by Lemma 11 of \cite{hitting_times}, with probability at least $1-\delta$ we have
\be \label{eq:hitting_time}
\tau_{\hat{X}^{(k)}}((U_k^{\varepsilon \hat{H}_k})_{\omega/\lambda}) &\leq \frac{4 \log(e^{2R\tilde{\lambda}\xi_k + d\mathrm{log}(\frac{2R}{r})}/\delta)}{\Phi_Y^2(\hat{\mathcal{K}}^{(k)}\backslash (U_k^{\varepsilon \hat{H}_k})_{\omega/\lambda})}\\
&\stackrel{{\scriptsize \textrm{Eq. }} \eqref{eq:conductance2}}{\leq} \frac{8R\tilde{\lambda}\xi_k + 4d\mathrm{log}(\frac{2R}{r}) + 4\log(\frac{1}{\delta})}{\left(\frac{1}{1536R}\sqrt{\eta_k/d} e^{-\frac{100d}{\varepsilon} \left[\frac{\alpha}{1-\alpha^{\dagger}}\left(3+\varepsilon \mathfrak{B} + \frac{\beta }{\mathfrak{D}}\right) + \frac{\beta}{\mathfrak{D}} \right]\log(R/\min(\frac{\varepsilon}{2\lambda} \mathfrak{D}, r'))}\right)^2}\\
& = \frac{8R\tilde{\lambda}\xi_k + 4d(\log(\bar{\xi}) + \log(\frac{2R\tilde{\lambda}}{\delta})) + 4\log(\frac{1}{\delta})}{\left(\frac{1}{1536R}\sqrt{\eta_k/d} e^{-\frac{100d}{\varepsilon} \left[\frac{\alpha}{1-\alpha^{\dagger}}\left(3+\varepsilon \mathfrak{B} + \frac{\beta }{\mathfrak{D}}\right) + \frac{\beta}{\mathfrak{D}} \right]\log(R/\min(\frac{\varepsilon}{2\lambda} \mathfrak{D}, r'))}\right)^2}\\
& \leq \frac{8R\tilde{\lambda}\xi_k + 4d(\log(1+ \bar{\xi}) + \log(\frac{2R\tilde{\lambda}}{\delta})) + 4\log(\frac{1}{\delta})}{\left(\frac{1}{1536R}\sqrt{\eta_k/d} e^{-\frac{100d}{\varepsilon} \left[\frac{\alpha}{1-\alpha^{\dagger}}\left(3+\varepsilon \mathfrak{B} + \frac{\beta }{\mathfrak{D}}\right) + \frac{\beta}{\mathfrak{D}} \right]\log(R/\min(\frac{\varepsilon}{2\lambda} \mathfrak{D}, r'))}\right)^2}\\
& \leq \frac{8R\tilde{\lambda}\xi_k + 4d(1+ \log(1+ \bar{\xi}) + \log(\frac{2R\tilde{\lambda}}{\delta})) + 4\log(\frac{1}{\delta})}{\left(\frac{1}{1536R}\sqrt{\bar{\eta}/d} e^{-\frac{100d}{\varepsilon} \left[\frac{\alpha}{1-\alpha^{\dagger}}\left(3+\varepsilon \mathfrak{B} + \frac{\beta }{\mathfrak{D}}\right) + \frac{\beta}{\mathfrak{D}} \right]\log(R/\min(\frac{\varepsilon}{2\lambda} \mathfrak{D}, r'))}\right)^2}\\
& \leq \frac{8R\tilde{\lambda}\xi_k + 4d(1+ \log(1+ \bar{\xi}) + \log(\frac{2R\tilde{\lambda}}{\delta})) + 4\log(\frac{1}{\delta})}{\left(\frac{1}{1536R}\sqrt{\bar{\eta}^{\dagger}/d} e^{-\frac{150d}{\varepsilon} \left[\frac{\alpha}{1-\alpha^{\dagger}}\left(3+\varepsilon \mathfrak{B}' + \frac{\beta }{\mathfrak{D}}\right) + \frac{\beta}{\mathfrak{D}} \right]\log(R/\min(\frac{\varepsilon}{2\lambda} \mathfrak{D}, r'))- \frac{75}{\varepsilon}\alpha\log(i_\mathrm{max}+1)}\right)^2}\\
& \leq \frac{8R\tilde{\lambda}\xi_k + 4d(1+ \log(1+ \bar{\xi}) + \log(\frac{2R\tilde{\lambda}}{\delta})) + 4\log(\frac{1}{\delta})}{\left(\frac{1}{1536R}\sqrt{\bar{\eta}^{\dagger}/d} e^{-\frac{150d}{\varepsilon} \left[\frac{\alpha}{1-\alpha^{\dagger}}\left(3+\varepsilon \mathfrak{B}' + \frac{\beta }{\mathfrak{D}}\right) + \frac{\beta}{\mathfrak{D}} \right]\log(R/\min(\frac{\varepsilon}{2\lambda} \mathfrak{D}, r'))}\right)^2}  \times  (i_\mathrm{max}+1)^{\frac{150}{\varepsilon}\alpha}\\
&\leq  i_{\mathrm{max}},\\
\ee
where the first equality is true since $r = \frac{\delta}{\bar{\xi} \tilde{\lambda}}$, the fourth inequality is true by the definition of $\bar{\eta}$,  the fifth inequality is true by the definition of $\bar{\eta}^{\dagger}$, and the last inequality is true by our choice of $i_{\mathrm{max}}$.

But by Equation \eqref{eq:delta3}, $X_i^{(k)} = \hat{X}_i^{(k)}$ with probability at least $1-5\delta$.
Therefore, since Equation \eqref{eq:hitting_time} holds with probability at least $1-\delta$, we have that
\be \label{eq:hitting_time2}
\tau_{X^{(k)}}((U_k^{\varepsilon \hat{H}_k})_{\omega/\lambda})  \leq i_{\mathrm{max}}.
\ee
with probability at least $1-6\delta$.

\noindent
Therefore, by Equation \eqref{eq:hitting_time2}, with probability at least $1-6\delta$ for some $0\leq i_k^\circ \leq i_{\mathrm{max}}$ we have $X_{i_k^\circ}^{(k)} \in (U_k^{\varepsilon \hat{H}_k})_{\omega/\tilde{\lambda}}$ and hence that
\be
F(X_{i_k^\circ}^{(k)}) \leq \varepsilon \hat{H}_k + \tilde{\lambda} \times \frac{\omega}{\tilde{\lambda}} =  \varepsilon \hat{H}_k + \varepsilon \mathfrak{D} \leq 2\varepsilon \hat{H}_k
\ee
and therefore, since $0\leq \alpha <1$,
\be
\frac{1}{5}\tilde{f}(x_{k+1}) \stackrel{{\scriptsize \textrm{Lemma }} \ref{lemma:compare}}{\leq} \hat{F}(x_{k+1}) = \min_{0\leq i \leq i_{\mathrm{max}}} \hat{F}(X_{i}^{(k)}) \leq \hat{F}(X_{i_k^\circ}^{(k)}) \leq 2F(X_{i_k^\circ}^{(k)}) + \beta \leq 4\varepsilon \hat{H}_k + \beta \leq 5\varepsilon \hat{H}_k.
\ee
Hence, for every $0\leq k \leq k_{\mathrm{max}}$ we have
\be \label{eq:contraction2}
\tilde{f}(x_{k+1})  = H_{k+1} \leq 25\varepsilon\hat{H}_k = 25\varepsilon \max(H_k, \mathfrak{D})
\ee
with probability at least $1-6\delta$.

\noindent
Therefore, by induction on Equation \eqref{eq:contraction2}, for every $0\leq k \leq k_{\mathrm{max}}$, we have
\be  \label{eq:contraction}
H_{k+1} \leq 25\varepsilon \times \max\left((25\varepsilon)^k H_0, \mathfrak{D}\right)
\ee
with probability at least $1-6\delta(k+1)$.

\noindent
By Lemma \ref{lemma:compare}, we have $k_{\mathrm{max}}= \lceil \frac{\log(5 J_0/\mathfrak{D})}{\log(\frac{1}{25 \varepsilon})}\rceil+1 \geq \lceil \frac{\log(H_0/\mathfrak{D})}{\log(\frac{1}{25 \varepsilon})}\rceil+1$. Then, with probability at least $1-6\delta (k_\mathrm{max}+1)$,
\be \label{eq:tilda}
\tilde{f}(x_{k_\mathrm{max}}) -F(x^\star)&=\tilde{f}(x_{k_\mathrm{max}})\\
&= H_{k_{\mathrm{max}}}\\
&\stackrel{{\scriptsize \textrm{Eq. }} \eqref{eq:contraction}}{\leq} 25\varepsilon \times \max \bigg((25\varepsilon)^{k_{\mathrm{max}}-1} H_0, \mathfrak{D}\bigg)\\
&\leq 25\varepsilon \times \mathfrak{D}\\
&\leq \mathfrak{D},
\ee
since $0\leq \varepsilon < \frac{1}{25}$ implies that $0\leq 25\varepsilon <1$.

\noindent
Hence,
\be
F(x_{k_\mathrm{max}}) -F(x^\star)&=F(x_{k_\mathrm{max}})\\
&\leq \frac{1}{1-\alpha}(\tilde{f}(x_{k_\mathrm{max}})+\beta)\\
&\leq \frac{1}{1-\alpha}(\mathfrak{D} + \beta),
\ee
 where the first equality holds since $F(x^\star)=0$.
\end{proof}

\subsection{The non-smooth case}
\label{sec:smoothing}
In this section we bound the gradient, supremum, and smoothness of the smoothed function $f_\sigma$ obtained from $F$ (Propositions \ref{lemma:smooth_gradient} and \ref{max_value} and Lemma \ref{lemma:Hessian}), where $f_\sigma$ is defined in Equation \eqref{eq:smoother}. We also bound the noise $|F(x)-f_\sigma(x)|$ of $f_\sigma$ (Lemma \ref{lemma:noise_smooth}).  We use these bounds in Section \ref{sec:proof_of_main_result} to Prove our main result (Theorem \ref{thm:main}).

\begin{proposition} \label{lemma:smooth_gradient} \textbf{(Gradient bound for smoothed oracle)}

\noindent For every $x\in \mathcal{K}$ we have
\be
\|\nabla \tilde{f}_\sigma(x)\| \leq  \frac{\sqrt{2d}}{\sigma} (2 \lambda R(1+ 2\alpha) + 2\beta).
\ee
\end{proposition}
\begin{proof}
\be
\|\nabla \tilde{f}_\sigma(x)\| &\leq \mathbb{E}_Z\left[\frac{1}{\sigma^2} \|Z\|  \left|\hat{F}(x+Z)- \hat{F}(x)\right|\right]\\
&\leq \mathbb{E}_Z\left[\frac{1}{\sigma^2} \|Z\|  \max_{y_1,y_2\in \mathcal{K}} |\hat{F}(y_2) - \hat{F}(y_1)| \right] \\
&\leq \frac{1}{\sigma}\max_{y_1,y_2\in \mathcal{K}} |\hat{F}(y_2) - \hat{F}(y_1)| \mathbb{E}_Z\left[\frac{1}{\sigma} \|Z\| \right]\\
&\leq \frac{1}{\sigma}\max_{y_1,y_2\in \mathcal{K}} |\hat{F}(y_2) - \hat{F}(y_1)| \mathbb{E}_Z\left[\frac{1}{\sigma} \|Z\| \right]\\
&= \frac{1}{\sigma}\max_{y_1,y_2\in \mathcal{K}} |\hat{F}(y_2) - \hat{F}(y_1)| \sqrt{2} \frac{\Gamma(\frac{d+1}{2})}{\Gamma(\frac{d}{2})}\\
&\leq \frac{\sqrt{2d}}{\sigma}\max_{y_1,y_2\in \mathcal{K}} |\hat{F}(y_2) - \hat{F}(y_1)|\\
&\leq \frac{\sqrt{2d}}{\sigma} (2 \lambda R(1+ 2\alpha) + 2\beta),
\ee
where the equality is true since $\frac{1}{\sigma} \|Z\|$ has $\chi$ distribution with $d$ degrees of freedom, and the second-to-last inequality is true since $\frac{\Gamma(\frac{d+1}{2})}{\Gamma(\frac{d}{2})} \leq \sqrt{d}$.  The last ineqaulty is true because $F$ is $\lambda$-Lipschitz, and because of our assumption on the noise (Equation \eqref{eq:assumption}).
\end{proof}

\begin{proposition}\label{max_value}  (maximum value of non-smooth noisy oracle)
For every $x \in \mathcal{K}_{\mathsf{r}}$, we have
\be
\hat{F}(x) \leq (1+\alpha) 2\lambda (R+ \mathsf{r}) + \beta.
\ee
\end{proposition}
\begin{proof}
Since $F$ is $\lambda$-Lipschitz,
\be \label{eq:a1}
F(x) \leq 2\lambda (R+\mathsf{r}) \quad \quad \forall x \in \mathcal{K}_{\mathsf{r}}.
\ee

\noindent
Thus,
\be
\hat{F}(x) \leq (1+\alpha) |F(x)| + \beta \leq (1+ \alpha) 2\lambda (R+\mathsf{r}) + \beta.
\ee
\end{proof}

\noindent
We recall the following Lemma from \cite{hitting_times}:
\begin{lemma} \label{lemma:Hessian} \textbf{(Lemma 17 in \cite{hitting_times})}
Suppose that $\hat{M}>0$ is a number such that $0\leq \hat{F}(x) \leq \hat{M}$ for all $x\in \mathcal{K}_{\mathsf{r}}$  then
\begin{enumerate}
\item \be
\mathbb{E}_Z[g_Z(x)] = \nabla \tilde{f}_\sigma(x) \quad \quad \forall x \in \mathcal{K}.
\ee
\item For every $u \in \mathbb{R}^d$, \be\mathbb{E}_Z[e^{\langle u, g_Z(x)\rangle(2\hat{M}/\sigma)^2}] \leq e^{\frac{4\hat{M}^2}{\sigma^2}\|u\|^2}.
\ee
\item \be \|\nabla^2 \tilde{f}_\sigma(x) \|_{\mathrm{op}} \leq \frac{2\hat{M}}{\sigma^2}.
\ee
\end{enumerate}

\end{lemma}

\noindent
We show that the smoothed gradient is a good approximation of $F$ for sufficiently small $\sigma$:
\begin{lemma}\label{lemma:noise_smooth} \textbf{(Noise of smoothed oracle)}
Let $A \subseteq A_t \subseteq \mathcal{K}_{\mathsf{r}}$ for some $A\subseteq \mathcal{K}_{\mathsf{r}}$ and some  $t>0$. Let $H' = \sup_{y \in A} F(y)$.
Then 
\be
|\tilde{f}_\sigma(x) - F(x)| \leq \lambda  \sigma(1+\alpha) \sqrt{d} + H'\times  e^{-\frac{t^2/\sigma^2-d}{8}} + \alpha H' + \beta
\ee
for every $x \in A$.
\end{lemma}

\begin{proof}
Define $\mathsf{N}(x) := \hat{F}(x) - F(x)$.  For any function $h: \mathbb{R}^d \rightarrow \mathbb{R}$, define
\be
\tilde{h}_\sigma(x) := \mathbb{E}_Z[h(x+Z)],
\ee where $Z \sim \mathcal{N}(0, \sigma^2 I_d)$.
Then for every $x \in A$ we have,
\be
|\tilde{f}_\sigma(x)-F(x)| &= |\tilde{F}_\sigma(x) + \tilde{\mathsf{N}}_\sigma(x) - F(x)|\\
&\leq |\tilde{F}_\sigma(x) - F(x)| + |\tilde{\mathsf{N}}_\sigma(x)|\\
&= \mathbb{E}_Z[|F(x+Z) - F(x)|] + |\mathbb{E}_Z[\mathsf{N}(x+Z)]|\\
&\leq \mathbb{E}_Z[|F(x+Z) - F(x)|] +  \mathbb{E}_Z[\alpha(H'+\lambda \|Z\|) + \beta] \\
&\leq \mathbb{E}_Z[\lambda \|Z\|] + H'\times \mathbb{P}(\|Z\|\geq t)  +  \mathbb{E}_Z[\alpha(H'+\lambda \|Z\|) + \beta] \\
&=\lambda(1+\alpha)\mathbb{E}_Z[\|Z\|] + H'\times \mathbb{P}(\|Z\|\geq t)  + \alpha H' + \beta\\
&=\lambda  \sigma(1+\alpha)\mathbb{E}_Z[\frac{1}{\sigma}\|Z\|] + H'\times  \mathbb{P}(\frac{1}{\sigma}\|Z\|\geq \frac{t}{\sigma})  + \alpha H' + \beta\\
&\leq\lambda  \sigma(1+\alpha)\sqrt{d} + H'\times  \mathbb{P}(\frac{1}{\sigma}\|Z\|\geq \frac{t}{\sigma})  + \alpha H' + \beta\\
&\leq \lambda  \sigma(1+\alpha) \sqrt{d} + H'\times  e^{-\frac{t^2/\sigma^2-d}{8}} + \alpha H' + \beta,\\
\ee
where the second inequality holds because $F$ is $\lambda$-Lipschitz on $\mathcal{K}_{\mathsf{r}}$ and also since $F$ is defined to be zero outside $\mathcal{K}_{\mathsf{r}}$ 
with $x\in \mathcal{K} \subseteq\mathcal{K}_{\mathsf{r}}$. The third inequality holds by our assumption on the noise (Equation \eqref{eq:assumption}), and since $F$ is defined to be zero outside $\mathcal{K}_{\mathsf{r}}$ . 
 The fourth inequality holds because $\frac{1}{\sigma}\|z\|$ is $\chi$-distributed with $d$ degrees of freedom. 
  The last inequality holds by the Hanson-Wright inequality (see for instance \cite{hanson1971bound}, \cite{rudelson2013hanson}).
\end{proof}

\subsection{Rounding the domain of the Markov Chain}
We now show that our constraint set $\hat{\mathcal{K}}$ is sufficiently ``rounded".  This roundness property is used to show that the Markov chain does not get stuck for a long time in corners of the constraint set.
\begin{lemma} \label{lemma:contraint_round} \textbf{(Roundness of constraint set)}
Let $\zeta_{\mathrm{max}} =  (\frac{r'}{10\sqrt{2}(d+20)})^2$. Let $\hat{\mathcal{K}} \subseteq \mathcal{K}'$ be a convex set.  Then for any $\zeta \leq \zeta_{\mathrm{max}}$ and any $x \in \hat{\mathcal{K}}_{r'}$ the random variable $W \sim \mathcal{N}(0, I_d)$ satisfies
\be
\mathbb{P}(\sqrt{2\zeta}W + x \in \hat{\mathcal{K}}_{r'}) \geq \frac{1}{3}.
\ee
\end{lemma}
\begin{proof}
Without loss of generality, we may assume that $x$ is the origin and that $\hat{\mathcal{K}}_{r'}$ contains the ball $B(a,r')$ where $a=(r', 0,\ldots, 0)^\top$  (since $\hat{\mathcal{K}}_{r'} = \hat{\mathcal{K}}+B(0,r')$  implies that there is a ball contained in $\hat{\mathcal{K}}_{r'}$ that also contains $x$ on its boundary. 
 We can then translate and rotate $\hat{\mathcal{K}}_{r'}$ to put $x$ and $a$ in the desired position).

Since $\mathbb{P}(\frac{1}{10} \leq W_1 \leq 100) \geq 0.45$, with probability at least $0.45$ we have that
\be
\frac{1}{10} \leq W_1 \leq 100
\ee
but $\zeta_{\mathrm{max}} = (\frac{r'}{10\sqrt{2}(d+20)})^2$, and hence, with probability at least $0.45$,
\be
\frac{\sqrt{2\zeta}}{r'}(d + 20) \leq W_1 \leq \frac{r'}{\sqrt{2\zeta}}.
\ee
But our choice of $\zeta_{\mathrm{max}}$ implies that $\sqrt{\frac{(r')^2}{2\zeta}}  - \frac{1}{\sqrt{\frac{(r')^2}{2\zeta}}}(d + 20) >0$, implying that
\be
\frac{r'}{\sqrt{2\zeta}}- \left(\sqrt{\frac{(r')^2}{2\zeta}} - \frac{1}{\sqrt{\frac{(r')^2}{2\zeta}}}(d + 20) \right) \leq W_1 \leq \frac{r'}{\sqrt{2\zeta}}+ \left(\sqrt{\frac{(r')^2}{2\zeta}} - \frac{1}{\sqrt{\frac{(r')^2}{2\zeta}}}(d + 20) \right).
\ee
But for any $\mathsf{a}>0$, we have $\sqrt{\mathsf{a}}- \frac{t}{\sqrt{\mathsf{a}}} \leq \sqrt{\mathsf{a}-t}$ for every $t\in [0,\mathsf{a})$, which implies
\be
\frac{r'}{\sqrt{2\zeta}}- \sqrt{\frac{(r')^2}{2\zeta} - (d + 20)} \leq W_1 \leq \frac{r'}{\sqrt{2\zeta}}+ \sqrt{\frac{(r')^2}{2\zeta} - (d + 20)}.
\ee
Therefore
\be
r'- \sqrt{(r')^2 - 2\zeta(d + 20)} \leq \sqrt{2\zeta}W_1 \leq r'+ \sqrt{(r')^2 - 2\zeta(d + 20)}.
\ee
Hence,
\be
(\sqrt{2\zeta}W_1 -r')^2 \leq (r')^2 - 2\zeta(d + 20),
\ee
which implies that
\be \label{eq:roundness1}
(\sqrt{2\zeta}W_1 -r')^2 + 2\zeta(d + 20) \leq (r')^2.
\ee
But by the Hanson-Wright inequality
\be \label{eq:Hanson}
\mathbb{P}(\sum_{j=2}^d W_j^2  \geq d + 20) \leq e^{-\frac{21}{8}} < \frac{1}{10}.
\ee
Thus, Equations \eqref{eq:Hanson} and \eqref{eq:roundness1} imply that with probability at least $0.45 -\frac{1}{10} \geq \frac{1}{3}$ we have
\be
\|\sqrt{2\zeta}W-a\| &= \|\sqrt{2\zeta}W-(r', 0,\ldots, 0)^\top\|^2\\
&= (\sqrt{2\zeta}W_1 -r')^2 + 2\zeta \sum_{j=2}^d W_j^2\\
&\leq (\sqrt{2\zeta}W_1 -r')^2 + 2\zeta(d + 20)\\
&\stackrel{{\scriptsize \textrm{Eq. }} \eqref{eq:roundness1}}{\leq}  (r')^2,
\ee
implying that $W \in B(a,r') \subseteq \hat{\mathcal{K}}_{r'}$ with probability at least $\frac{1}{3}$. 
\end{proof}

\subsection{Proof of Main Result (Theorem \ref{thm:main})} \label{sec:proof_of_main_result}
In this section, we prove Theorem \ref{thm:main}.  We do so by applying the bounds on the smoothness of $f_\sigma$  of Section \ref{sec:smoothing} to Theorem \ref{thm:error}. 

We note that in this section we will use ``$\alpha$"  and ``$\beta$" exclusively to denote the multiplicative and additive noise levels of $F$.  We will then set the smooth oracle $\tilde{f}$ to be $\tilde{f} = \tilde{f}_\sigma$, where $\tilde{f}_\sigma$ is the smooth function obtained from $F$,  defined in Equation  \eqref{eq:smoother}. As an intermediate step in proving the main result, we show that $\tilde{f}_\sigma$ has multiplicative noise level $2\alpha$ and additive noise level $2\beta$.

\begin{proof}
We will assume that $\alpha<\frac{1}{800}$.  This assumption is consistent with the statment of Theorem \ref{thm:main}, which assumes that $\alpha = O(1)$. 

Define the following constants: $M = 2 \lambda R + 2 \beta$, $\hat{M}= 6 \lambda R + \beta$, $L= \frac{4 \hat{M}}{\sigma^2}$, $G= \frac{2\hat{M}}{\sigma}$, $b_{\textrm{max}}=1$, $\zeta_{\textrm{max}} =  \left(\frac{r'}{10\sqrt{2}(d+20)}\right)^2$, and $\tilde{\lambda} = \frac{\sqrt{2d}}{\sigma} (2 \lambda R(1+ 2\alpha) + 2\beta)$.

We set $\sigma = \frac{1}{2} \min\left(\frac{\beta}{\lambda(1+\alpha)\sqrt{d}}, \frac{\mathsf{r}}{\sqrt{\log(\frac{1}{\alpha})+d}}\right)$.  Recall from Section \ref{sec:smoothed_from_nonsmooth} that $\sigma$ determines the amount of smoothness in $\tilde{f}_\sigma$.  A larger value of $\sigma$ means that $\tilde{f}_\sigma$ will be smoother, decreasing the running time of the algorithm. On the other hand, a smaller value of $\sigma$ means that $\tilde{f}_\sigma$ will be a closer approximation to $F$, and consequently lead to a lower error.  We choose $\sigma$ in such a way so that the error $\hat{\varepsilon}$  will be bounded by the desired value $\hat{\varepsilon}$.

Set parameters of Algorithms \ref{alg:SGLD} and \ref{alg:annealing} as follows: 
\begin{itemize}
\item Fix $\varepsilon = \frac{1}{50}$.

\item Let $\mathfrak{D} = \frac{2}{3}\hat{\varepsilon}$.

\item Define $J_0 := \hat{F}(x_0)$ and set the number of epochs to be
\be\label{eq:k_max}
k_{\mathrm{max}}= \left\lceil \frac{\log(5 J_0/\mathfrak{D})}{\log(2)}\right\rceil+1.
\ee

\item For every $0\leq k \leq k_{\mathrm{max}}$, let $J_k := \hat{F}(x_k)$, and define $\hat{J}_k := \max(J_k, \mathfrak{D})$.

\item Fix $\delta = \frac{\delta'}{6( k_{\mathrm{max}}+1)}$.

\item At every $k \geq 0$, set the temperature to be
\be \label{eq:xi_k}
\xi_k =  \frac{4 d \log(R/\min(\frac{\varepsilon}{2\lambda} \mathfrak{D}, r'))}{\frac{1}{5} \varepsilon \hat{J}_k}.
\ee
  Define $\bar{\xi} := \frac{4 d \log(R/\min(\frac{\varepsilon}{2\lambda}  \mathfrak{D}, r'))}{\frac{1}{25}\varepsilon  \mathfrak{D}}$.

\item Set $r = \frac{\delta}{\bar{\xi}\tilde{\lambda}}$.

\item Define
\be
\bar{\eta}^{\dagger} := c\min \left\{\zeta_{\mathrm{max}}, d\frac{\omega^2}{\lambda^2}, \frac{b_{\mathrm{max}}^2}{d}, \frac{1}{Rd^3((\bar{\xi} G)^2 +\bar{\xi} L)^2}\right\}
\ee
and
\be
\mathfrak{B}':= \frac{\left(d\log(2\frac{R}{r}) +\delta + 1+ \log(\frac{1}{\delta})\right))}{2 d \log(R/\min(\frac{\varepsilon}{2\lambda} \mathfrak{D}, r'))}.
\ee

 \item Set the number of steps $i_{\mathrm{max}}$ for which we run the the Markov chain $X$ in each epoch to be
 \be \label{eq:i_max}
 \hspace{-12mm} i_\mathrm{max} = \left \lceil \left ( \frac{8R\tilde{\lambda}\xi_k + 4d(1+ \log(1+ \bar{\xi}) + \log(\frac{2R\tilde{\lambda}}{\delta})) + 4\log(\frac{1}{\delta})}{\left(\frac{1}{1536R}\sqrt{\bar{\eta}^{\dagger}/d} e^{-\frac{150d}{\varepsilon} \left[\frac{\alpha}{1-\alpha^{\dagger}}\left(3+\varepsilon \mathfrak{B}' + \frac{\beta }{\mathfrak{D}}\right) + \frac{\beta}{\mathfrak{D}} \right]\log(R/\min(\frac{\varepsilon}{2\lambda} \mathfrak{D}, r'))}\right)^2} \right )^{\frac{1}{1-\frac{150}{\varepsilon}\alpha}} \right \rceil +1.
 \ee
\item Define $\mathfrak{B}:= \frac{\left(d\log(2\frac{R}{r}) +\delta +  \log(i_{\max}+1) + \log(\frac{1}{\delta})\right))}{2 d \log(R/\min(\frac{\varepsilon}{2\lambda} \mathfrak{D}, r'))}$.

\item For every $\xi >0$ define
\be \label{eq:eta_k}
\hspace{-2mm} \eta(\xi) := c\min \left\{\zeta_{\mathrm{max}}, d\frac{\omega^2}{\lambda^2}, \frac{b_{\mathrm{max}}^2}{d}, \frac{(e^{-\frac{100d}{\varepsilon} \left[\frac{\alpha}{1-\alpha^{\dagger}}\left(3+\varepsilon \mathfrak{B} + \frac{\beta}{\mathfrak{D}}\right) + \frac{\beta}{\mathfrak{D}} \right]\log(R/\min(\frac{\varepsilon}{2\lambda} \mathfrak{D}, r'))})^2}{Rd^3((\xi G)^2 +\xi L)^2}\right\},
\ee
\noindent where $\omega = \varepsilon \mathfrak{D}$, and $c$ is the universal constant in Lemma 15 of \cite{hitting_times}.
\noindent We set the step size at each epoch $k$ to be $\eta_k = \eta(\xi_k)$.  We also define $\bar{\eta} := \eta(\bar{\xi})$.

\item Set $D = \sqrt{2\bar{\eta} d}$.

\end{itemize}

\noindent We determine the constants for which $\tilde{f}= \tilde{f}_\sigma$ satisfies the various assumptions of Theorem \ref{thm:error}.

Since $\sigma  = \frac{1}{2} \min\left(\frac{\beta}{\lambda(1+\alpha)\sqrt{d}}, \frac{\mathsf{r}}{8\sqrt{\log(\frac{1}{\alpha})+d}}\right)$, by Lemma \ref{lemma:noise_smooth}, we have that
\be
|\tilde{f}_\sigma(x) - F(x)| \leq 2\alpha F(x) + 2\beta \quad \quad \forall x \in \mathcal{K}.
\ee
So, with a slight abuse of notation, we may state that $\tilde{f} = \tilde{f}_\sigma$ has multiplicative noise of level $2\alpha$ and additive noise of level $2\beta$, if we use ``$\alpha$" and ``$\beta$" to denote the noise levels of $\hat{F}$.

 Hence, $M = 2 \lambda R + 2 \beta \geq \sup_{x\in \mathcal{K}} \tilde{f}_\sigma(x)$.  
  By Lemma \ref{lemma:contraint_round}, part 1 of Assumption \ref{assumption:A} is satisfied with constant $\zeta_{\textrm{max}}$.  
 By Lemma \ref{lemma:Hessian} and Proposition \ref{max_value},  $\tilde{f}_\sigma$ satisfies parts 2 and 3 of Assumption \ref{assumption:A} with constants $L$, $G$ and $b_{\textrm{max}}$, (recall that we defined these constants at the beginning of this proof).  
  By Proposition \ref{lemma:smooth_gradient},  $\|\nabla \tilde{f}_\sigma(x) \| \leq \tilde{\lambda}$ for all $x \in \mathcal{K}$. 
 Therefore, applying Theorem \ref{alg:annealing} with the above constants and the smoothed function $f_\sigma$, we have, 
\be
 F(\hat{x}) -F(x^\star) &\leq \frac{1}{1-2\alpha}(\mathfrak{D} + 2\beta) \leq \hat{\varepsilon},
\ee
with running time that is polynomial in $d$, $e^{\frac{8d}{\varepsilon} \left[\frac{\alpha}{1-\alpha^{\dagger}}\left(3+\varepsilon \mathfrak{B}' + \frac{\beta}{\mathfrak{D}}\right) + \frac{\beta}{\mathfrak{D}} \right] \log(R/\min(\frac{\varepsilon}{2\lambda} \mathfrak{D}, r'))}$, $R$, ${\lambda}$, ${\tilde{\lambda}}$, $L$, G, $\zeta_{\mathrm{max}}$, $b_{\mathrm{max}}$, and $\log(\frac{1}{\delta})$. This completes the proof of the Theorem.
\end{proof}

	\bibliographystyle{plain}
\bibliography{annealing} 
	
\end{document}